\documentclass[12pt,draftclsnofoot,onecolumn]{IEEEtran}

\usepackage{blindtext}
\usepackage{amssymb}
\usepackage{amsmath}
\usepackage{graphicx}
\usepackage{cite}
\usepackage{citesort}
\usepackage{multicol}
\usepackage{amsfonts}
\usepackage{color}
\usepackage{subfigure}
\usepackage{epstopdf}
\usepackage{amssymb}
\usepackage{comment}
\usepackage{amsmath}
\allowdisplaybreaks[4]
\usepackage{booktabs}
\usepackage{multirow}
\usepackage{bmpsize}
\usepackage{algorithm}
\usepackage{algorithmic}

\usepackage{lipsum}
\usepackage{stfloats}
\newtheorem{theorem}{\textbf{Theorem}}

\newtheorem{lemma}{\textbf{Lemma}}

\newtheorem{corollary}{\textbf{Corollary}}

\newtheorem{proof}{\textbf{Proof}}

\makeatletter
\def\ScaleIfNeeded{%
\ifdim\Gin@nat@width>\linewidth \linewidth \else \Gin@nat@width
\fi } \makeatother

\begin{document}
%\pagestyle{fancyplain}
%
%\pagestyle{fancy}
%\lhead[]%
 %   {\footnotesize Physical layer security}
%\cfoot{}

\title{\huge{Spatial-spectral Terahertz Networks}}
\author{Zheng Lin, Lifeng Wang, Bo Tan, and Xiang Li
\thanks{Z. Lin, L. Wang, and X. Li are with the Department of Electrical Engineering, Fudan University, Shanghai, China (E-mail: lifengwang@fudan.edu.cn).}
\thanks{B. Tan is with the Faculty of Information Technology and Communication Sciences, Tampere University, Finland (E-mail: $\rm bo.tan@tuni.fi$).}
}

\maketitle

\begin{abstract}
This paper focuses on the spatial-spectral terahertz (THz) networks, where transmitters equipped with  leaky-wave antennas send information to their receivers at the THz frequency bands. As a directional and nearly planar antenna, the leaky-wave antenna allows for information transmissions with narrow beams and high antenna gains. The conventional large antenna arrays are confronted with challenging issues such as scaling limits and path discovery in the THz frequencies. Therefore, this work exploits the potential of leaky-wave antennas in the dense THz networks, to establish low-complexity THz links. By addressing the propagation angle-frequency coupling effects, the transmission rate is analyzed. The results show that the leaky-wave antenna is efficient for achieving the high-speed  transmission rate. The co-channel interference management is unnecessary when the THz transmitters with large subchannel bandwidths are not extremely dense. A simple subchannel allocation solution is proposed, which enhances the transmission rate compared with the same number of subchannels with the equal allocation of the frequency band. After subchannel allocation, a low-complexity power allocation method is proposed to improve the energy efficiency.
\end{abstract}


%========================================================================
\section{Introduction}

{The emerging services such as edge computing, immersive communication and tactile internet demand high-speed data transmissions. In an attempt to achieve these services, large millimeter wave (mmWave)  frequency bands have been leveraged in 5G (currently 24-50 GHz \cite{3gpp_104}).} However,  with the increasing numbers of smart devices and autonomous vehicles in the Internet-of-Things, more bandwidths are always required to strengthen the ultra-reliable and low latency communications (URLLC).  Since the terahertz (THz) frequency bands are abundant, THz communication is viewed as a promising 6G technology for reaching unprecedentedly high data rates~\cite{Koenig_THz_2013,ted_6G_2019,Rikkinen_THz_2020}.

In order to counteract the severe path losses in the higher frequencies, large antenna arrays with narrow beams  have been adopted in 5G and its evolution~\cite{sunshu_2014}. In light of sparse-scattering mmWave channel environments, hardware costs and power consumptions etc.,  analog beamforming or hybrid beamforming approaches are recommended in the mmWave systems~\cite{sunshu_2014,Ayach_2014}. The implementations of antenna arrays in the THz communications are investigated in \cite{cenlin_2016,zhi_chen2019,Bile_2019,youli_2020,chonghanSept_2020}, where various beamforming/precoding designs are proposed. However, using conventional antenna arrays has encountered many challenging issues in the THz frequencies, e.g., array architecture for accommodating highly dense THz antennas~\cite{ted_6G_2019}, phase modulation technique for developing THz phase shifters~\cite{tiejuncui_2020,Ge_zhang2017}, and path discovery~\cite{Yasaman_2020}. In addition,  efficient feed network designs~\cite{D_headland2018} and beam squint mitigation~\cite{A_J_Seeds1995} may also be required for developing wideband THz phased arrays. Therefore, it is essential to appropriately choose the THz antennas without adding significant link budgets and power consumptions~\cite{Rikkinen_THz_2020,zhi_chen2019}.

The leaky-wave antenna enables the wave to travel along the guiding architecture, in order to intensify the radio energy in preferred directions~\cite{Oliner_2007,Jackson_2008,Jackson_2012}. As a low-cost and easy-to-manufacture traveling-wave antenna, the leaky-wave antenna can provide frequency-dependent narrow beams with high antenna gains~\cite{Jackson_2012} and frequency-scanning capability~\cite{BAHL_1975,Gezhang2019}. It has been utilized in many areas including  frequency-division multiplexing THz communication~\cite{FDM_2015}, THz radar sensing~\cite{Murano2017,Brown2020,Monnai2020},  and physical layer security enhancement~\cite{Security2020}. {Different from the beam management with conventional large arrays in which an exhaustive search for the best beam is required~\cite{PoleseMag}, the extensive beam training in the THz link discovery with leaky-wave antennas is unnecessary~\cite{Yasaman_2020,Ghasempour_LeakyTrack}.}  One key feature of the leaky-wave antenna is that it enables the information transmissions in a spatial-spectral manner, i.e., frequencies are correlated with the transmission directions. There are also other directional antenna designs such as horn and lens antennas. In particular, horn antenna has been used in the mmWave and THz channel measurements~\cite{TED2015TCOM,ted_6G_2019} and fixed wireless access for THz communication systems~\cite{Koenig_THz_2013}. Unlike leaky-wave antenna, these antennas are inherently fixed beam solution with a single direction, and integrating them for expanding coverage needs to be properly addressed~\cite{Rikkinen_THz_2020}.

The aforementioned works only study the case of point-to-point THz communication with leaky-wave antenna~\cite{FDM_2015,Security2020}. When there exist large numbers of transceivers in the THz networks with massive connections, the co-channel interference has an adverse effect on the transmission rate, which has to be evaluated. Moreover, THz transmissions utilize much larger frequency bandwidths, and  subchannel allocation plays an essential role in controlling the number of subchannels while guaranteeing the targeted performance, to keep the peak-to-average power ratio (PAPR) at a low level. Unfortunately, few research contributions investigate the subchannel allocation when applying the leaky-wave antenna in the THz communications. In addition, energy efficiency enhancement is of importance, to reduce power consumption. Therefore, we adopt the leaky-wave antennas to harness the THz waves in the dense THz networks, and the main contributions are concluded as follows:
\begin{itemize}
\item \textbf{THz Networks with Leaky-wave Antennas:} In the considered THz networks, each transmitter equipped with a leaky-wave antenna (TE$_1$ mode) sends information to its corresponding receiver with an omnidirectional antenna. {Since none-line-of-sight (NLoS) links heavily depend on the reflectors in the THz frequencies~\cite{ted_6G_2019,Khalid_2016,J_Ma_2018} and the interference from NLoS links in dense THz networks can be negligible similar to the dense mmWave networks~\cite{T_Bai_2015}, we focus on the dominant line-of-sight (LoS) links.} As a useful tool to evaluate the performance behavior in large-scale wireless networks~\cite{Haenggi2009}, stochastic geometry is employed to model the spatial distributions of transmitters.
\item \textbf{Average Transmission Rate and Subchannel Allocation:} In light of the spatial energy distribution under the leaky-wave antenna radiation, the average transmission rate is quantified for an arbitrary THz subchannel. Considering the fact that different THz frequencies have different  subchannel bandwidths and undergo distinct channel conditions, a closed-form solution for subchannel allocation is proposed to enhance the transmission rate with the minimum number of subchannels. Then, a low-complexity power allocation is designed to maximize the average energy efficiency (EE) over the number of subchannels.

\item \textbf{Design Insights:} Our results show that high-speed transmission rates are achievable in the dense THz networks with leaky-wave antennas, and noise-limited phenomenon occurs when the transmitters are not super dense. The average transmission rate can vary dramatically with slightly different values of aperture length or attenuation coefficient. The proposed subchannel allocation improves the average transmission rate in the leaky-wave antenna systems, compared with the same number of subchannels with the equal allocation of the frequency band. Our power allocation method can improve the EE. It is demonstrated that center frequencies with the maximum radiated energy may not be the best option when ignoring the channel gains. The slight increase in the attenuation coefficient can dramatically improve the average transmission rate for large frequency bandwidths,  since more frequencies are in the main-lobe that captures large radiated energy.
\end{itemize}

The rest of this paper is organized as follows. The considered system model is described in Section~\ref{System_description},  and
the average transmission rate is analyzed based on stochastic geometry in Section~\ref{stability_section}.
Subchannel allocation is determined in Section \ref{sec:velocity_handover}. Section~\ref{sec:EE} focuses on the EE enhancement.
Section~\ref{sec:simulation} covers the simulation results. Finally, some concluding remarks are presented in Section~\ref{conclusion_section}.

\section{System Descriptions}\label{System_description}
\begin{figure}[t!]
\centering
\includegraphics[width=4.4 cm]{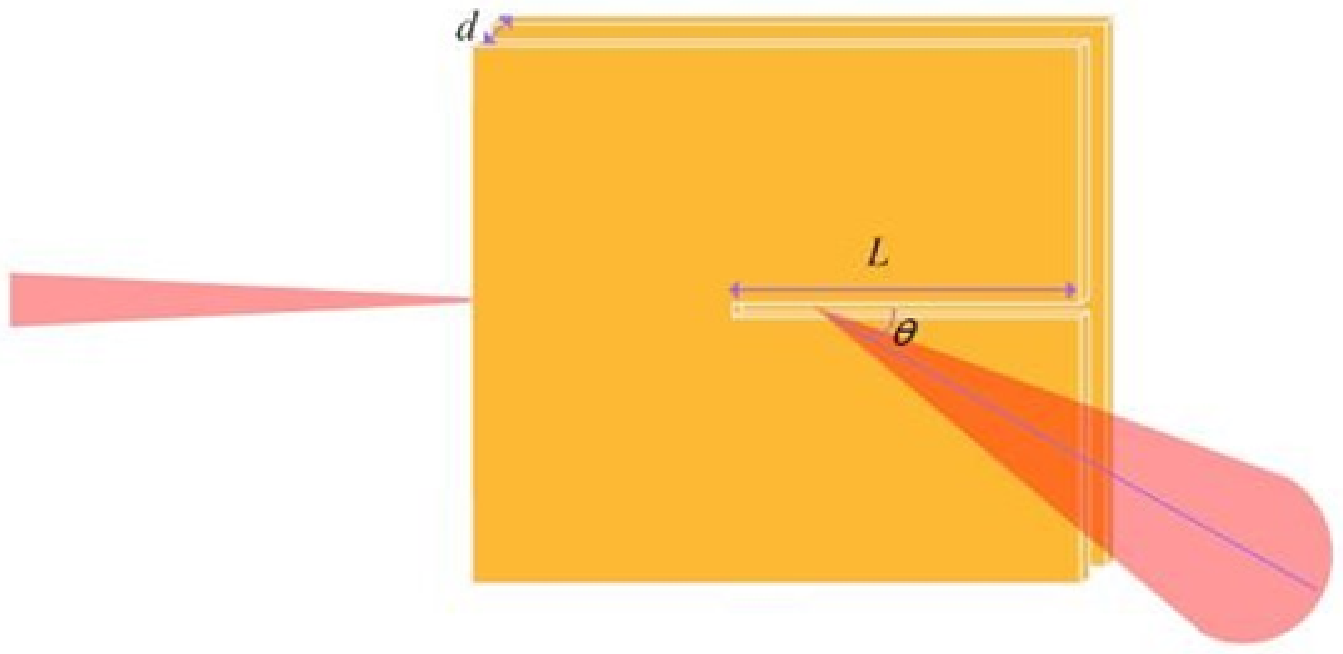}
\caption{An illustration of leaky-wave antenna with the lowest transverse-electric (TE$_1$) mode.
}% \textcolor[rgb]{1.00,0.00,0.00}{edge clouds  at or linked to the SBSs through optical fiber}
\label{leaky_wave}
\end{figure}

\begin{figure}[t!]
\centering
\includegraphics[width=6 cm,height=6 cm]{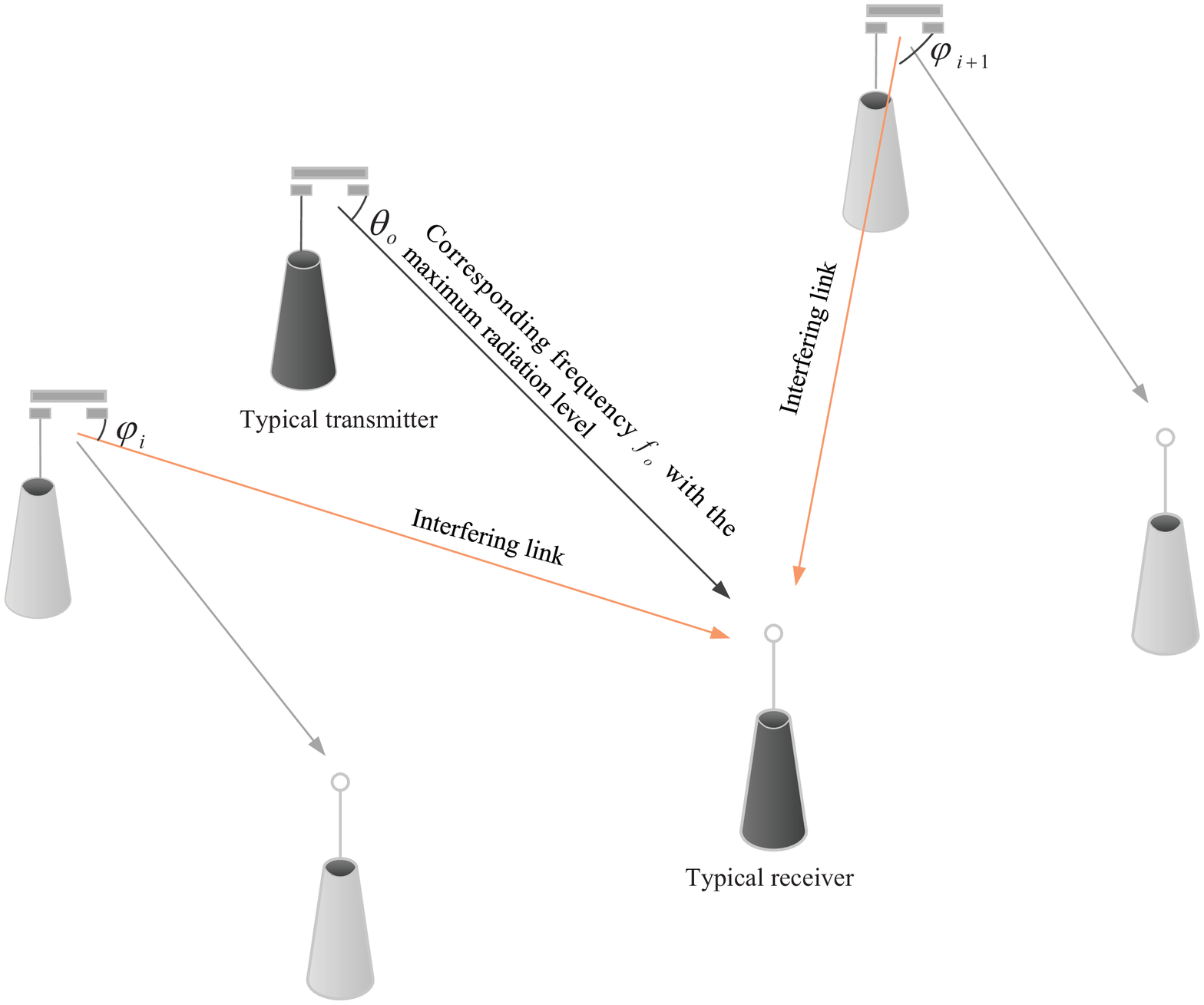}
\caption{An illustration of THz network with leaky-wave antennas, where the interfering antenna gains depend on the propagation angles and the frequencies of the interfering links.
}% \textcolor[rgb]{1.00,0.00,0.00}{edge clouds  at or linked to the SBSs through optical fiber}
\label{network_model}
\end{figure}

As shown in Fig.~\ref{leaky_wave}, the guiding structure of a 1D leaky-wave antenna consists of a rectangular waveguide with a longitudinal slit, where $d$, $L$ and $\theta$ denote the inter-plate distance, aperture length and propagation angle, respectively. In the large-scale THz networks illustrated by Fig.~\ref{network_model}, each transmitter equipped with a leaky-wave antenna communicates with its corresponding receiver equipped with an omnidirectional antenna (namely single-input single-output){\footnote{This work can be easily extended to the case of both transmitters and receivers equipped with leaky-wave antennas where the effective end-to-end antenna gains are considered~\cite{Ghasempour_LeakyTrack}.}}, and they are randomly located following a homogeneous Poisson point process (PPP) $\Phi _{\rm THz}$ with density $\lambda_{\rm THz}$.

Given a THz frequency $f$ and the propagation angle $\theta\left(0 < \theta < 90^{o}\right)${\footnote{{The propagation angle can be estimated by using the method of~\cite{Yasaman_2020} without extensive beam training.}}}, the far-field radiation pattern of the considered leaky-wave antenna is given by~\cite{Sutinjo2008}
\setcounter{equation}{0}\begin{align}\label{antenna_gain}
G\left(f,\theta\right) = L \mathrm{sinc}\left[\left(-j\alpha-k_0 \cos\theta+\beta \right)\frac{L}{2}\right],
\end{align}
where $j=\sqrt{-1}$, $\alpha$ is the attenuation coefficient resulting from the power absorption in the structure, $k_0=2\pi f/c$ with the speed of light $c$ is the wavenumber of the free-space, $\beta = k_0 \sqrt{1-\left(\frac{f_{\rm co}}{f}\right)^2}$  is the phase constant of the TE$_1$ mode based traveling wave, in which $f_{\rm co}=\frac{c}{2d}$ is the cutoff frequency~\cite{Mendis2010}. The available frequencies should be higher than the cutoff frequency (i.e., $f>f_{\rm co}$),  to enable that THz wave propagates away from the antenna structure, namely fast wave radiation~\cite{Sutinjo2008}. Given a LoS direction $\theta$ of a receiver, the maximum level of the radiation can be achieved by using the the following frequency~\cite{Jackson_2012,FDM_2015}:
\begin{align}\label{antenna_gain11}
f^{\max}\left(\theta\right) = \frac{f_{\rm co}}{\sin \theta}.
\end{align}
As shown in Fig. \ref{antenna_pattern},  the frequency for maximizing the radiation is reliant on the beam angle, and lowering attenuation coefficient results in a narrower beam.
\begin{figure}[ht]
     \centering
    \subfigure[$f=80$GHz, $\alpha=30$rad/m]{
         \centering
         \includegraphics[width=4.1 cm]{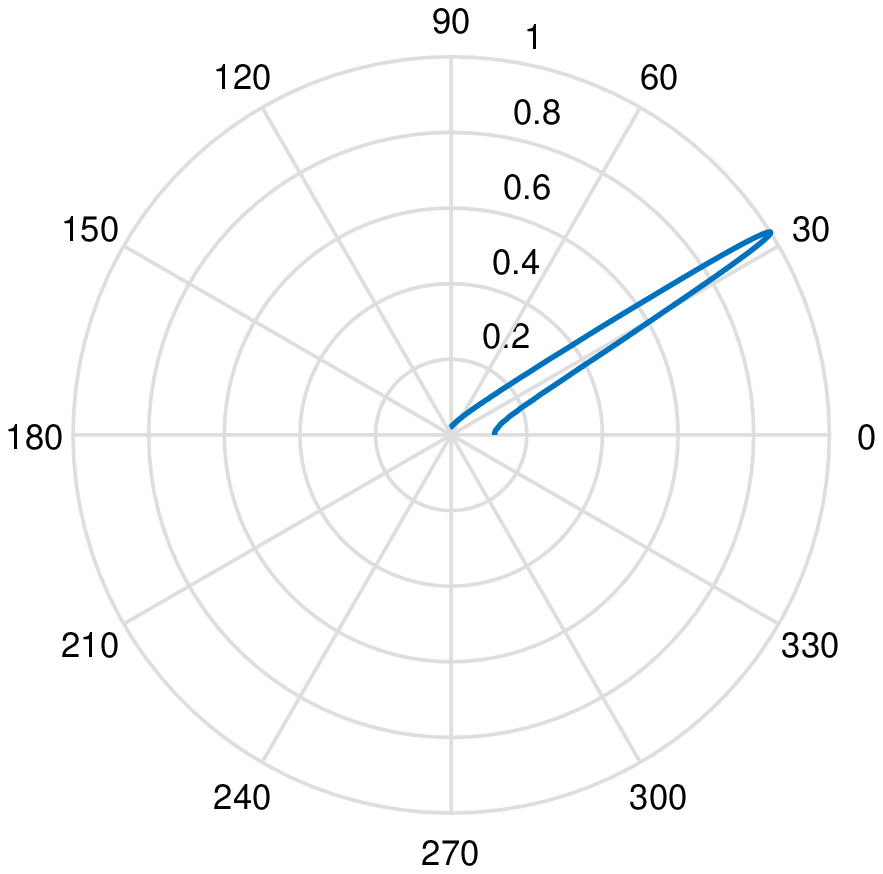}
      \label{fig1a}
     }
     \subfigure[$f=80$GHz, $\alpha=60$rad/m]{
         \centering
         \includegraphics[width=4.1 cm]{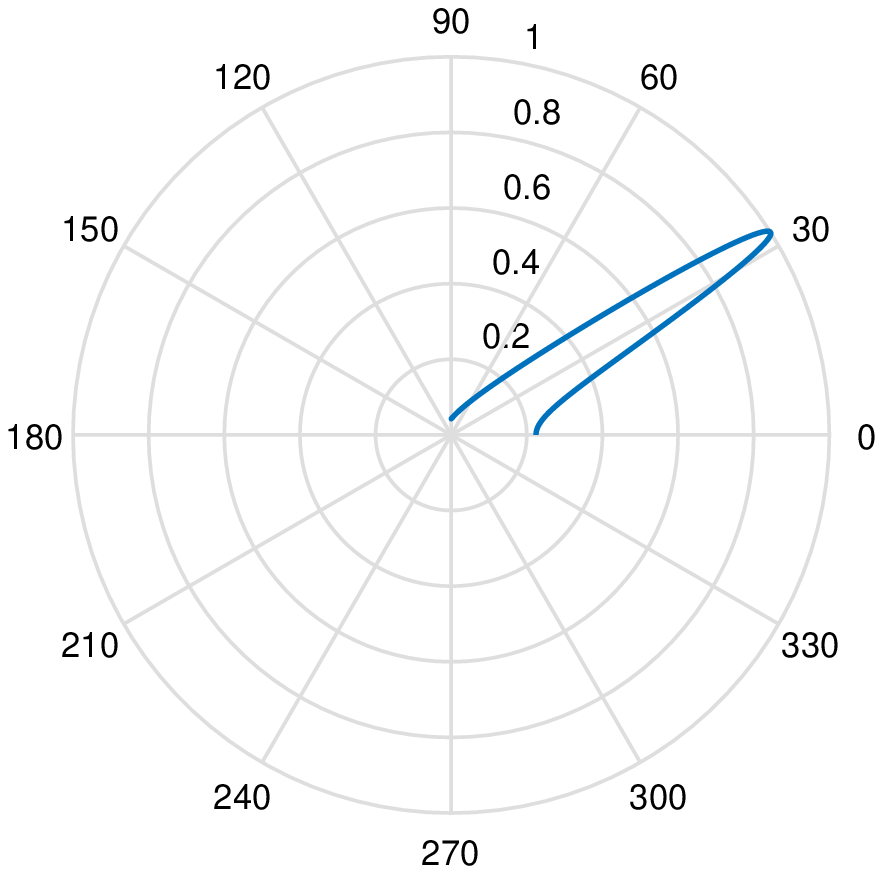}
      \label{fig1b}
    }

    \subfigure[$f=160$GHz, $\alpha=30$rad/m]{
         \centering
         \includegraphics[width=4.1cm]{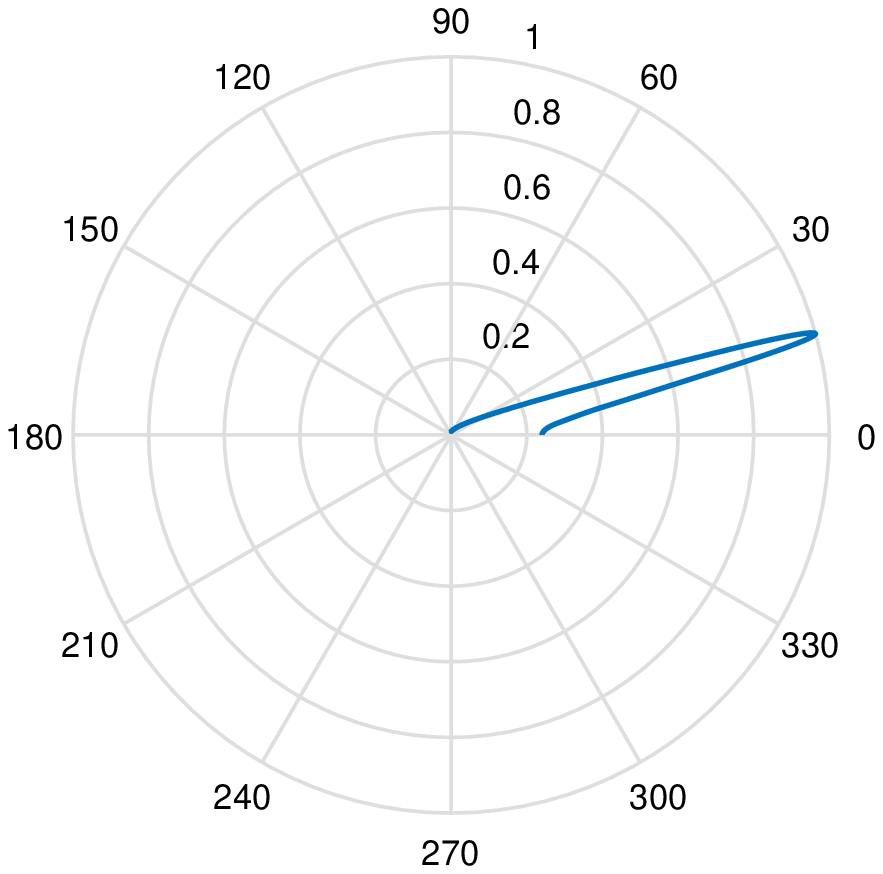}
         \label{fig1c}
     }
    \subfigure[$f=160$GHz, $\alpha=60$rad/m]{
         \centering
         \includegraphics[width=4.1 cm]{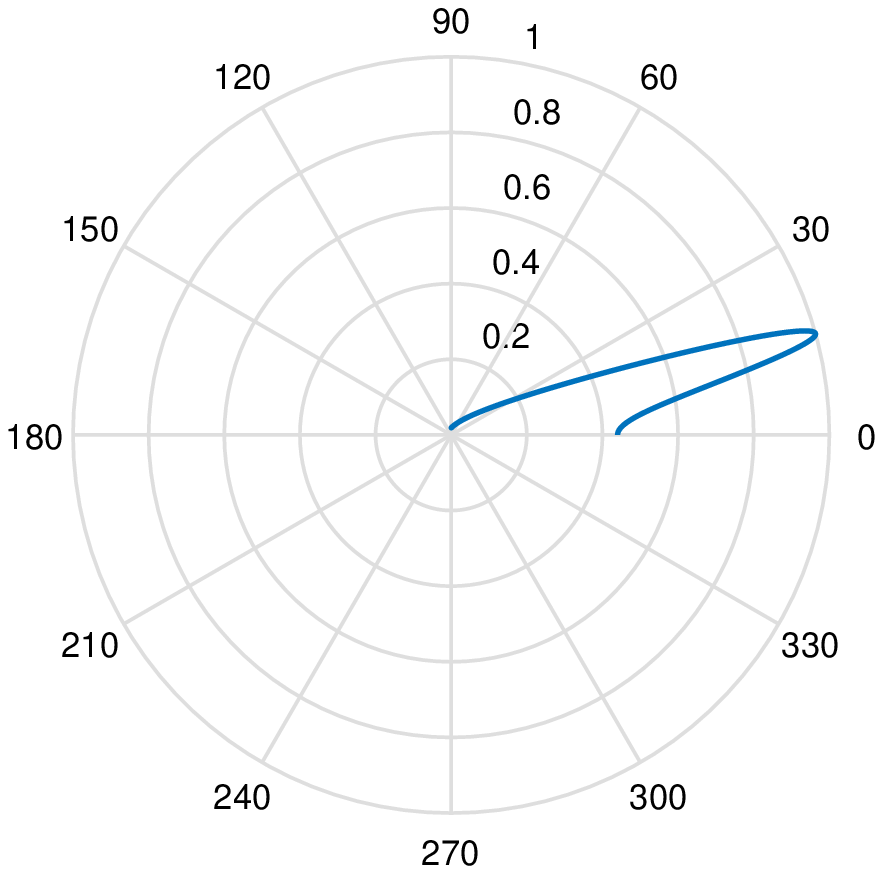}
         \label{fig1d}
     }
  \caption{Examples of the leaky-wave antenna's radiation pattern with $L=5.5$cm and $d=3.5$mm.}
 \label{antenna_pattern}
\end{figure}

Although transmitters may have different propagation angles and send information messages with the maximum antenna gains in the different THz frequencies based on \eqref{antenna_gain11}, co-channel interference still exists in the large-scale THz networks, where the nearby transceivers may use the same subchannels for the low propagation angle difference $\Delta \theta$.
As measured in~\cite{FDM_2015},  for a given $\Delta \theta$, its corresponding frequency bandwidth is
\begin{align}\label{bandwidth_radiation}
B\left(\theta\right)=\frac{f_{\rm co}\Delta \theta}{\sin \theta \tan \theta},
\end{align}
which can be interpreted as the range of the frequencies whose levels of radiation are close to the frequency $f^{\max}\left(\theta\right)$. In addition, different transmitters may have different cutoff frequencies and propagation angles,
however, they may leverage the same frequency band, as indicated in \eqref{antenna_gain11}. Therefore, the transmission rate of a subchannel at a typical receiver (at the origin)  can be expressed as
%\footnote{Although the value of $G\left(f,\theta\right)$ given by \eqref{antenna_gain} is not equal to the effective antenna gain, it can well predict the performance behavior of the leaky-wave antenna, as measured in \cite{Security2020}. }
\begin{align}\label{subchannel_rate}
R_o = B_o \log_2\left(1+
\frac{{q_t \widetilde{G}\left( {f_o,\theta _o } \right)\ell\left( r_o  \right)}}{{\sum\nolimits_{i \in \widetilde{\Phi} _{\rm THz} /o} {q_t \widetilde{G}\left( {f_i,\varphi _i } \right)\ell\left( r_i \right)}  + \sigma _o^2}}
\right),
\end{align}
where $B_o$ is the subchannel bandwidth, $q_t$ is the transmit power spectral density (PSD), $\widetilde{G}\left( {f,\theta } \right) = \xi G\left(f,\theta\right)$ with a constant $\xi$\footnote{Through measuring the effective antenna gain for a particular leaky-wave antenna structure, $\xi$ can be easily obtained and known a priori.} is the effective antenna gain,  $\theta_o$ is the typical propagation angle, $\ell\left( r \right)$ is the path loss function with the communication distance $r$,  $f_o$ and $f_i$ are the frequencies used by the typical transmitter and the $i$-th interferer, respectively, $\varphi _i$ is the interfering link direction from the $i$-th transmitter (interferer) to the typical receiver, which is assumed to be independently and uniformly distributed in $\left(0, \frac{\pi}{2}\right)$, $\sigma _o^2$ is the PSD of the noise. Here, $\ell\left( r \right)=\rho\left(f\right)\left(\max\left(D, r\right)\right)^{-\eta}$ with the intercept $\rho\left(f\right)={(\frac{{\text{c}}}{{4\pi {f}}})^2}$, reference distance $D$ and path loss exponent $\eta$~\cite{TED2015TCOM,andrew_model2017}. { In practice, THz coverage area is usually not much larger than the mmWave (the coverage radius of mmWave is about 200m~\cite{TED2013IEEE_Access,Xianghao_yu2017}), in such a limited THz coverage, the molecular absorption loss can be negligible compared to the high path loss~\cite{Joonas2019,Joonas2020}. Moreover, the THz path loss is more significant than the mmWave and thus THz co-channel interference mainly comes from the typical receiver's neighboring transmitters. Therefore, the effect of molecular absorption loss on the level of interference is also negligible.}  The small-scale fading effect is omitted since it is insignificant in LoS links by using directional antennas at the higher frequencies~\cite{Samimi2013}. It should be noted in \eqref{subchannel_rate} that the frequency $f_i$ solely depends on the link between the $i$-th transmitter and its corresponding receiver.

\section{Average Transmission Rate}\label{stability_section}
In this section, we evaluate the average transmission rate in the dense THz networks with leaky-wave antennas, which is an important performance indicator. Like mmWave,  THz communication is also susceptible to the blockage. There are empirical (e.g., 3GPP) and analytical (e.g., random shape theory) blockage models~\cite{andrew_model2017}, hence we consider the generalized case, i.e., the LoS probability function for a THz link at a distance $r$ is denoted as $P_{\rm LoS}\left(r\right)$. The average transmit rate is calculated as~\cite{Hamdi2008,yongxu2017}
\begin{align}\label{average_rate}
&\overline{R}_o = \mathbb{E}\left[R_o\right] \nonumber\\
&=\frac{B_o}{\ln 2}\int_0^\infty  {\frac{1}{s}} \left( {1 - e^{ - sY\left(\max\left(D, r_o\right)\right)^{-\eta _{\rm LoS}}} } \right) \Theta\left(s\right) e^{ - s\sigma _o^2 } ds,
\end{align}
where $\Theta\left(s\right)=\mathbb{E}\left[ {e^{ - sI} } \right]$,  $ \eta _{\rm LoS}$ is the LoS path loss exponent, $Y =q_t \widetilde{G}\left( {f_o,\theta _o } \right)\rho$, and the interference $I = \sum\nolimits_{i \in \widetilde{\Phi} _{\rm THz} /o} {q_t \widetilde{G}\left( {f_i,\varphi _i } \right)\ell\left( r_i \right)} $. The frequency of the typical transmission link is assumed to be $f_o=f^{\max}\left(\theta_o\right) $, i.e., main-lobe gain can be obtained at the typical receiver. It is obvious that co-channel interference occurs when the frequencies used by the interferers are in the range $\left[f_o-\frac{B_o}{2}, f_o+\frac{B_o}{2}\right]$. {We consider the worst-case scenario that all the transmitters are homogeneous (namely identical cutoff frequency). As such, there are more co-channel interfering links and the typical receiver is more likely to be covered by the main-lobes of the interferers' leaky-wave antennas. } {According to \eqref{bandwidth_radiation}, the probability that a transmitter uses the frequency band $B_o$ at $f_o$  is  given by
\begin{align}\label{int_direction}
P_{f_o} &=\int_0^{\Delta \theta _o } {\frac{2}{\pi }d\theta } \nonumber\\
&=\frac{2 B_o \sin \theta_o \tan \theta_o}{\pi f_{\rm co}}.
\end{align}
\begin{figure*}[!t]
\normalsize
\setcounter{equation}{10}\begin{align}\label{final_expression_rate}
&\bar R_o {\rm{ = }}\frac{{B_o }}{{\ln 2}}\int_0^\infty  {\frac{1}{s}} \left( {1 - e^{ - sY\left( {\max \left( {D,r_o } \right)} \right)^{ - \eta _{{\rm{LoS}}} } } } \right)\exp \left( { - s\sigma _o^2  - 2\pi \lambda _{{\rm{THz}}} P_{f_o } \int_0^\infty  {P_{{\rm{LoS}}} \left( r \right)} \left( {1 - \Xi \left( r \right)} \right)rdr} \right)ds \\
&\mathrm{with}~\Xi\left( r \right)~\mathrm{given~by}~\eqref{LPF_sub2}. \nonumber
\end{align}
\hrulefill
\end{figure*}
Based on the thinning theorem~\cite{Haenggi2009}, the density of the PPP $\widetilde{\Phi}_{\rm THz}$ is $\lambda_{\rm THz} P_{f_o} $. Since the interference caused by NLoS links in the dense network is negligible at the THz frequencies, the LoS interferers can be modeled as the non-homogeneous PPP with the density function $\lambda_{\rm THz} P_{f_o} P_{\rm LoS}\left(r\right)$, which is another important feature for large-scale THz networks with leaky-wave antennas.} Therefore, by using the Laplace functional of the PPP, $\Theta (s)$ can be evaluated as
\setcounter{equation}{6}\begin{align}\label{LPF}
\Theta (s) = \exp \left( - 2\pi \lambda _{\rm{THz}} P_{f_o} \int_0^\infty  {P_{{\rm{LoS}}} \left( r \right)} \left(1 - \Xi\left( r \right)  \right)rdr\right),
\end{align}
where $\Xi\left( r \right)$ is
\begin{align}\label{LPF_sub}
\Xi\left( r \right)=\mathbb{E}\left[e^{ - s q_t \widetilde{G}\left( {f_o,\varphi _i } \right) \rho (\max \{ D,r\} )^{ - \eta _{\rm LoS}} } \right].
\end{align}
To solve \eqref{LPF_sub}, we first need to determine the directions of the interfering links in the spatial-domain, which have a detrimental effect on the transmission rate. It is seen from \eqref{bandwidth_radiation} that the interfering link direction $\varphi _i$ meets the following condition:
\begin{align}\label{LPF_sub_1}
\varphi _i  \in \left[\theta_o-\frac{\Delta \theta_o}{2},\theta_o+\frac{\Delta \theta_o}{2}\right],
\end{align}
where $\Delta \theta_o = \frac{B_o\sin \theta_o \tan \theta_o }{f_{\rm co}}$. Based on \eqref{LPF_sub_1}, $\Xi\left( r \right)$ is explicitly given by
\begin{align}\label{LPF_sub2}
\Xi\left( r \right) =  \int_{\theta _o  - \Delta \theta _o /2}^{\theta _o  + \Delta \theta _o /2} {\frac{2}{\pi }e^{ - sq_t \widetilde{G}\left(f_0 ,\varphi\right)\rho (\max \{ D,r\} )^{ - \eta _{\rm LoS} } } } d\varphi.
\end{align}
Substituting \eqref{LPF} and \eqref{LPF_sub2} into \eqref{average_rate}, we can obtain the average transmission rate given by \eqref{final_expression_rate}.

Considering the fact that $G\left(f_0 ,\varphi\right) \leq G\left( {f_o,\theta_o} \right)$ as $\varphi \in \left[\theta_o-\frac{\Delta \theta_o}{2},\theta_o+\frac{\Delta \theta_o}{2}\right]$, the lower bound of the average transmission rate is
\setcounter{equation}{11}\begin{align}\label{lower_bound}
\hspace{-0.3cm}\overline{R}_o ^{\rm L} = &\frac{B_o}{\ln 2}\int_0^\infty  {\frac{1}{s}} \left( {1 - e^{ - sY\left(\max\left(D, r_o\right)\right)^{-\eta _{\rm LoS}}} } \right)e^{ - s\sigma _o^2 } \nonumber\\
&\quad \quad \;\;\; \times \exp \Bigg[ - 2\pi \lambda _{{\rm{THz}}} P_{f_o } \int_0^\infty  {P_{{\rm{LoS}}} \left( r \right)} \Big(1-\frac{2\Delta \theta_o}{\pi }  \nonumber\\
& \quad \quad \quad  \quad\qquad \;\;\;\ \times e^{ - sY (\max \{ D,r\} )^{ - \eta _{\rm LoS} } } \Big)rdr\Bigg] ds.
\end{align}

The average transmission rate given by \eqref{average_rate} is derived for an arbitrary subchannel. In practice, it is essential to properly determine the number of subchannels and each subchannel bandwidth for a large number of THz bandwidths. In the next section, we provide an efficient subchannel allocation solution.

\section{Subchannel Allocation}\label{sec:velocity_handover}
In the THz systems, the abundance of bandwidths needs to be divided into many subchannels, and the bandwidth of each subchannel depends on the effective antenna gain of the leaky-wave antenna and the channel condition. However, multi-carrier transmissions result in a PAPR issue and large number of subchannels could create higher PAPR~\cite{Seung2005,Guangshi2015}. The use of large mmWave bandwidths with high PAPR waveforms has already posed demanding power amplifier requirements in 5G systems~\cite{Shakib2017,feiwang2019}. Therefore, one of our aims is to determine the minimum number of subchannels for maximizing the transmission rate given the THz bandwidths, which is helpful to tune the number of subchannels for PAPR control. Moreover, the center frequency of each subchannel has to be appropriately selected. The reason is that frequency given by \eqref{antenna_gain11} for obtaining the maximum radiation energy may not be the best option since different THz frequencies have varying THz channel gains with significant path losses~\cite{ted_6G_2019}. In addition, not all available THz bandwidths can be applied under the quality of service (QoS) constraint when using the leaky-wave antenna.

As mentioned above, the considered subchannel allocation problem is formulated as
\begin{align}\label{SA_problem}
&\mathop {\max }\limits_{{\bf{B}},{\bf{f}}} \sum\limits_n {B_n \log _2 \left( {1 + \frac{\gamma _n \left( f_n  \right)} {\sigma_o^2 }} \right)}   \\
&\mathrm{s.t.} ~\mathrm{C1:}~\sum\limits_n {B_n }  \le B_{total} , \nonumber \\
&~\mathrm{C2:}~\bigcap\limits_n {\left\{ {f\left| {f \in \left[ {f_n  - \frac{{B_n }}{2},f_n  + \frac{{B_n }}{2}} \right]} \right.} \right\}}  = \emptyset , \nonumber\\
&~\mathrm{C3:}~\frac{\gamma _n \left( f_n  \right)} {\sigma_o^2 }  \ge \gamma _{\rm th} ,\;\,\;\forall n, \nonumber \\
&~\mathrm{C4:}~\left\| {\gamma _n \left( {f_n  - \frac{{B_n }}{2}} \right)\left| {_{\rm dB} } \right. - \gamma _n \left( {f_n  + \frac{{B_n }}{2}} \right)}\left| {_{\rm dB} } \right. \right\| \le \varepsilon ,\;\,\;\forall n, \nonumber \\
&~\mathrm{C5:}~B_n\geq 0,\;\,\;\forall n, \nonumber
\end{align}
where ${\bf{B}}=\left[B_n\right]$, ${\bf{f}}=\left[f_n\right]$, and $\gamma _n \left( {f_n } \right) = q_t \widetilde{G}\left(f_n ,\theta\right)\ell \left( {r} \right)$ is the receive PSD. Constraint $\mathrm{C1}$ describes the total available bandwidth $B_{total}$; $\mathrm{C2}$ avoids overlap between subchannels; $\mathrm{C3}$ is the QoS constraint with the threshold $\gamma _{\rm th}$; $\mathrm{C4}$ ensures that the received signal power (signal strength) difference is below a small value $\varepsilon$ in the frequencies of a subchannel; $\mathrm{C5}$ ensures that $B_n$ is non-negative value. In problem \eqref{SA_problem},  interference is negligible because of noise-limited THz networks, as confirmed in Section VI.
\begin{figure}[ht]
     \centering
    \subfigure[Radiation pattern, $\alpha=30$rad/m]{
         \centering
         \includegraphics[width=4.1 cm]{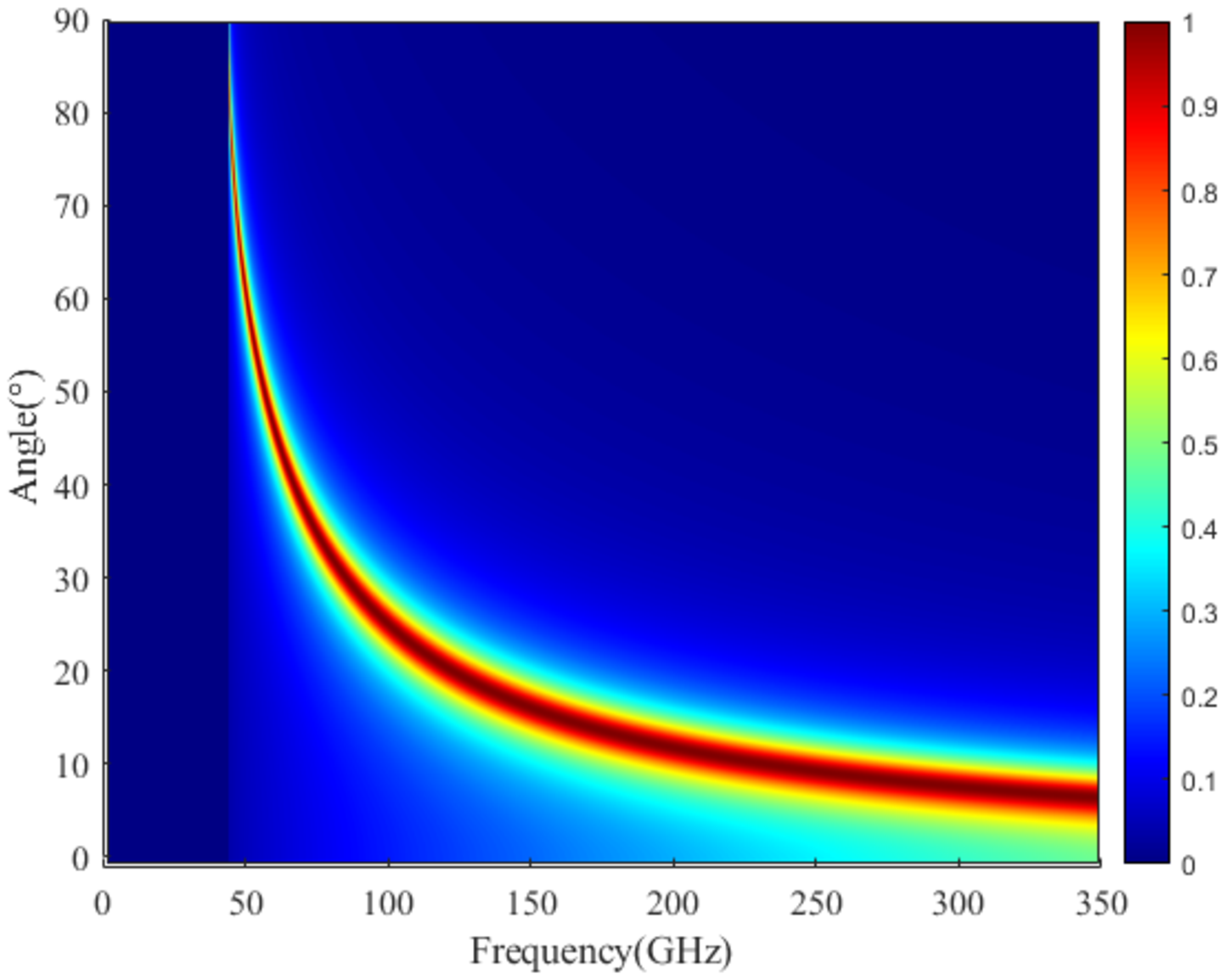}
      \label{fig1a}
     }
     \subfigure[Radiation pattern, $\alpha=60$rad/m]{
         \centering
         \includegraphics[width=4.1 cm]{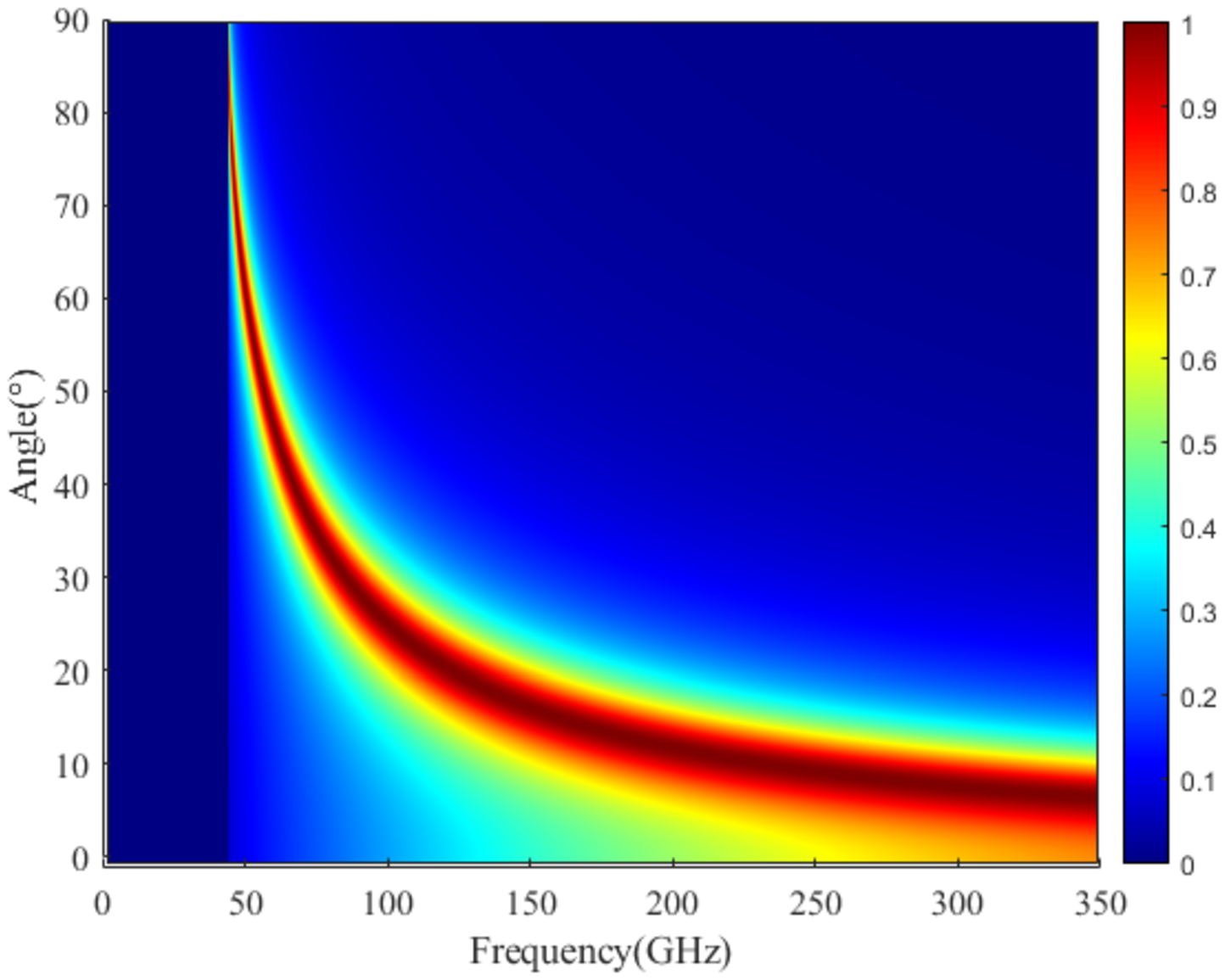}
      \label{fig1b}
    }

    \subfigure[Receive SNR with $\alpha=30$rad/m]{
         \centering
         \includegraphics[width=4.1cm]{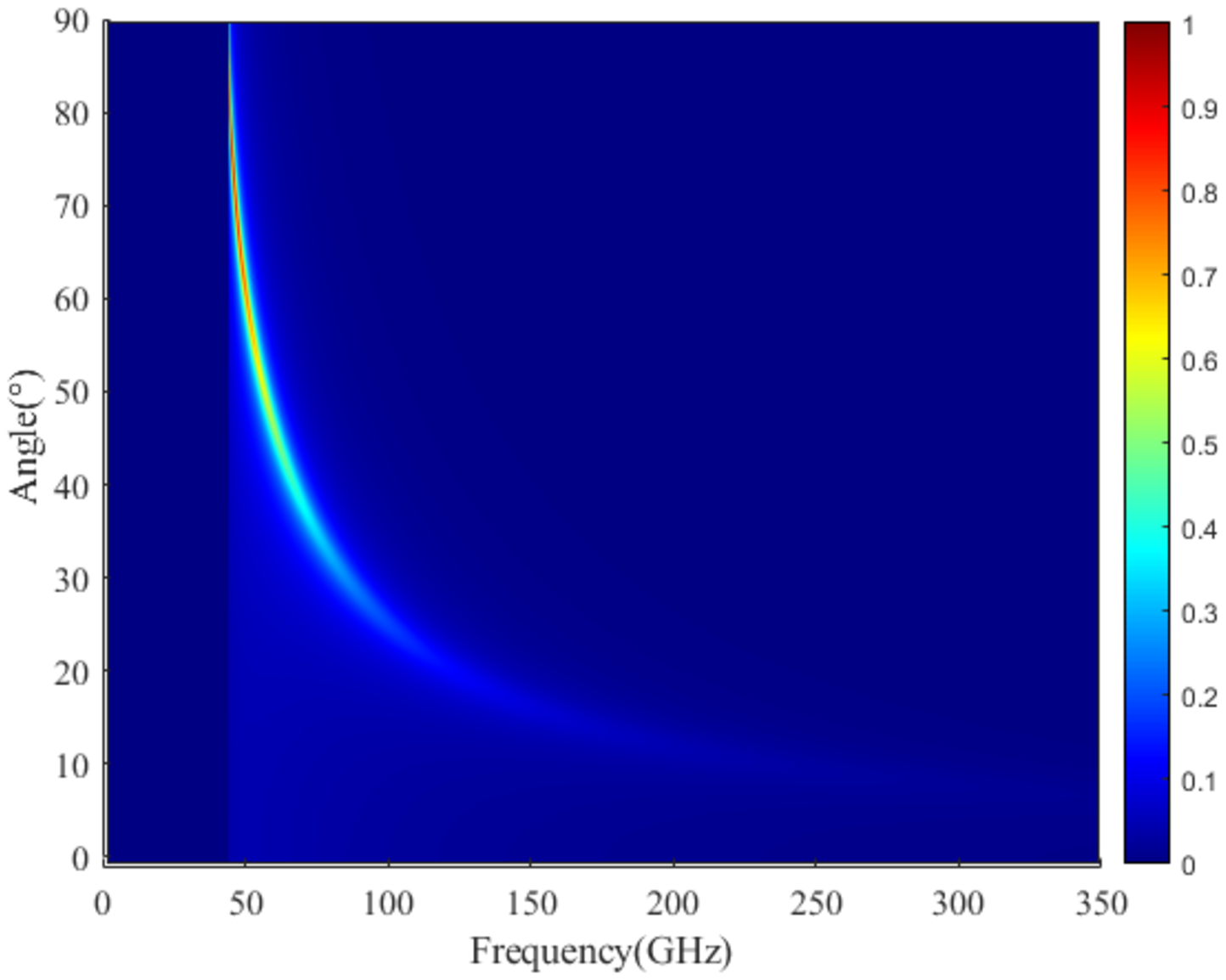}
         \label{fig1c}
     }
    \subfigure[Receive SNR with $\alpha=60$rad/m]{
         \centering
         \includegraphics[width=4.1 cm]{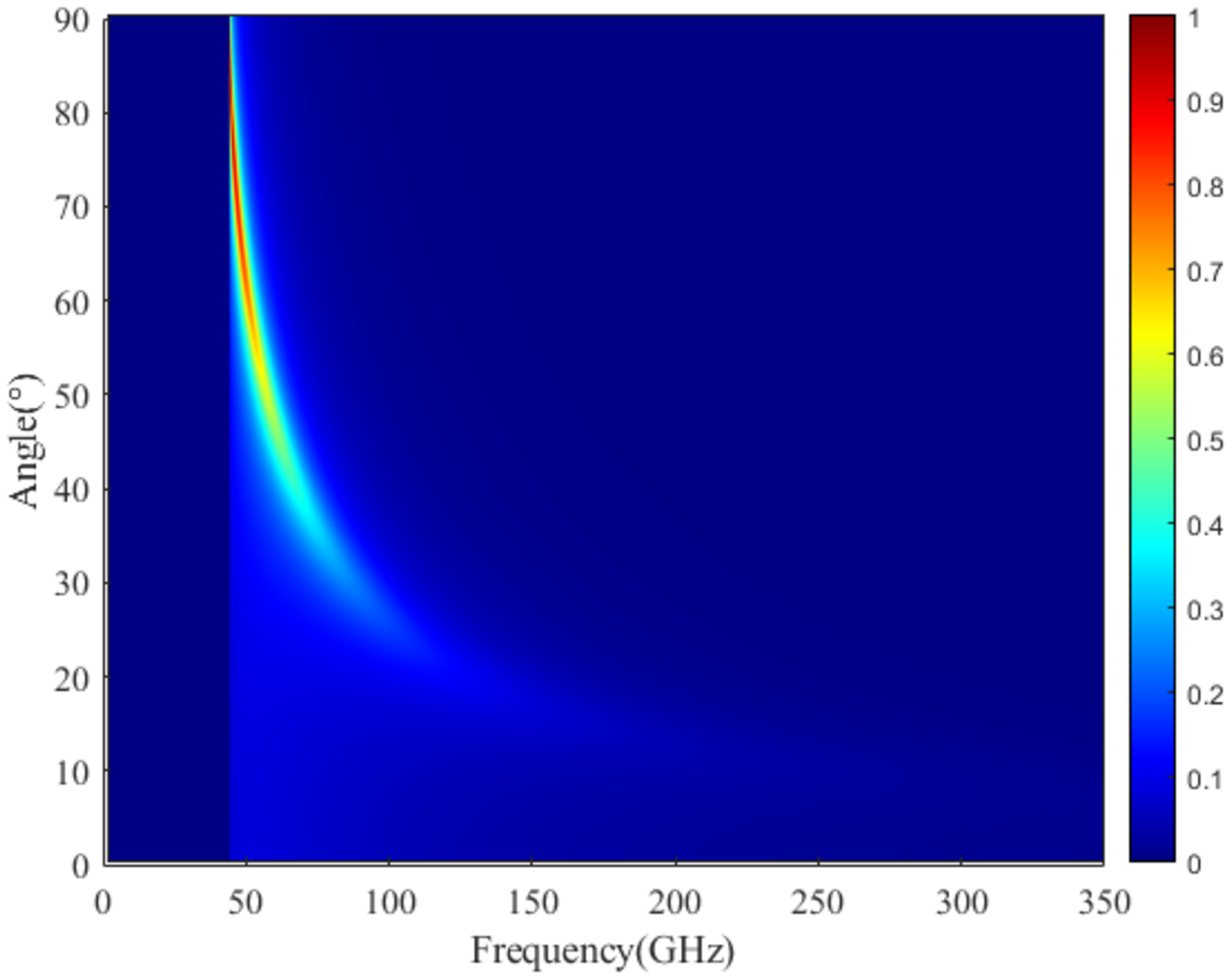}
         \label{fig1d}
     }
  \caption{{Heatmaps to illustrate the leaky-wave antenna's angle-frequency coupling behaviors in terms of the normalized radiation pattern and receive signal-to-noise ratio (SNR)}, where $L=5.5$cm, $d=3.5$mm, $q_t=-71.76$dBm/Hz, $\sigma_o^2=-168$dBm/Hz, the free-space path loss model is adopted with a fixed distance $r_o=50$m~\cite{Joonas2019,Joonas2020}. }
 \label{Angle-freq}
\end{figure}
Since the antenna gain cannot be explicitly calculated, solving the combinatorial problem \eqref{SA_problem} is  challenging. Thus, we provide a low-complexity approach to cope with it. One feature of using leaky-wave antenna is that a large range of THz frequencies (above 0.1 THz) have a nearly identical level of radiation and similar signal strength along a propagation angle, as seen in Fig. \ref{Angle-freq} (Similar result also seen in Fig. 10 of \cite{Yasaman_2020}).

Moreover, the recent THz channel modeling works~\cite{Joonas2019,Joonas2020} have demonstrated that {given the low range communication distance, pure free-space path loss (FSPL) model can well predict the channel of LoS THz link. It is also shown in~Fig. 5 that the molecular absorption loss is relatively marginal for the below 0.35 THz frequency bands, which are the potential 6G frequency bands~\cite{ted_6G_2019}.}
\begin{figure}[ht]
     \centering
         \includegraphics[width=3.0 in,height=2.4 in]{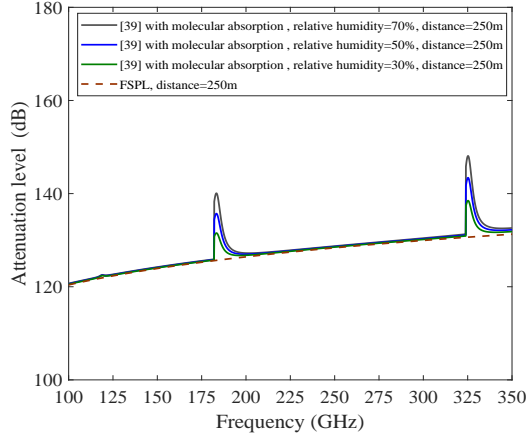}
{  \caption{Attenuation levels for the 100GHz-350GHz channels based on the models of~\cite{Joonas2019} and FSPL, respectively. It is seen that FSPL plays a dominant role and the molecular absorption losses with different humidity values are relatively marginal within 250m distance at below 0.35 THz frequency bands.}}
 \label{Attenuation}
\end{figure}
Therefore, based on FSPL model, constraint $\mathrm{C4}$ in \eqref{SA_problem}  can be rewritten as
\begin{align}\label{constraint_4}
 B_n  \le 2f_n \frac{{10^{\varepsilon /20}  - 1}}{{10^{\varepsilon /20}  + 1}}.
\end{align}
We note that different values of the attenuation coefficient $\alpha$ primarily influence the main-lobe and its nearby side-lobe levels, and have little effect on their trend in the frequency range of interest~\cite{Jackson_2008}, which is also shown in \cite{Sutinjo2008}. Thus we can approximate problem \eqref{SA_problem} by setting $\alpha=0$. In addition,  frequencies of interest need to be close to $f^{\max}\left(\theta\right)$, to achieve large radiated energy. As such, by applying the Taylor series expansion truncated to the first order at $f^{\max}\left(\theta\right)$, $G\left(f,\theta\right)$ with $\alpha=0$ in the frequency range of interest can be evaluated as
\begin{align}\label{G_taylor}
G\left(f,\theta  \right)\approx L-\frac{ L^3}{24}\left(\beta-k_0 \cos\theta\right)^2.
\end{align}

By considering \eqref{constraint_4}, \eqref{G_taylor} and $\beta \approx k_0\left(1-\frac{f_{\rm co}^2}{2f^2}\right) $, problem \eqref{SA_problem} is rewritten as
\begin{align}\label{SA_problem_re}
&\mathop {\max }\limits_{{\bf{B}},{\bf{f}}} \sum\limits_n B_n \log _2 \left( 1 + F_{\rm SNR}\left(f_n\right)\right)   \\
&\mathrm{s.t.} ~\mathrm{C1}, ~~\mathrm{C2},~~\mathrm{C5}, \nonumber\\
&~\mathrm{\widetilde{C}3:}~  F_{\rm SNR}\left(f_n\right) \ge \gamma _{\rm th} ,\;\,\;\forall n, \nonumber \\
&~\mathrm{\widetilde{C}4:}~\eqref{constraint_4},\;\,\;\forall n, \nonumber
\end{align}
where $F_{\rm SNR}\left(f_n\right)=\frac{q_t\xi\ell\left(r\right)} {\sigma_o^2 } \left(L-\frac{ L^3 k_0^2}{24}\left(1-\frac{f_{\rm co}^2}{2f_n^2}- \cos\theta\right)^2\right)$. To deal with the combinatorial constraint $\mathrm{C2}$, we propose a greedy-based solution.  We first find the best $f_1^*$ by solving the following subproblem
\begin{align}\label{SA_problem_re11}
&\mathop {\max }\limits_{f_1}  F_{\rm SNR}\left(f_1\right)  \\
&\mathrm{s.t.} ~\mathrm{\widetilde{C}3}, \nonumber
\end{align}
\noindent which means the best center frequency that has the largest signal strength is chosen first. Based on \eqref{SA_problem_re11}, we have
\begin{theorem}
The optimal solution of the subproblem \eqref{SA_problem_re11} is given by
\begin{align}\label{optimal_CF}
 f_1^*= \mathbf{1}\left(F_{\rm SNR}\left(f_1^{(1)}\right) \ge \gamma _{\rm th} \right) f_1^{(1)} ,
\end{align}
where $f_1^{(1)}=\frac{{f_{{\rm{co}}}^2 }}{{\sqrt {\frac{{12c^2 }}{{L^2 \pi ^2 }} + 2\left( {1 - \cos \theta _o } \right)f_{{\rm{co}}}^2 } }}$, and $\mathbf{1 }(A)$ represents the indicator
function that returns one if the condition $A$ is met.
\end{theorem}
\begin{proof}
See Appendix A.
\end{proof}
It is shown from \textbf{Theorem 1} that although the frequency given by \eqref{antenna_gain11} has the maximum radiated energy, it may not necessarily achieve the largest signal strength without considering the channel environments. Based on \textbf{Theorem 1} and problem \eqref{SA_problem_re}, the optimal bandwidth of the first subchannel with the center frequency $f_1^*$ is $B_1^*  = 2f_1^* \frac{{10^{\varepsilon /20}  - 1}}{{10^{\varepsilon /20}  + 1}}$. As such, we can obtain the center frequency of the $n$-th subchannel and the corresponding subchannel bandwidth as follows:
\begin{theorem}
The optimal center frequency of the $n$-th ($n \geq 2$) subchannel is given by
\begin{align}\label{optimal_CF_n}
 f_n^*= \mathbf{1}\left(F_{\rm SNR}\left(f_n\right) \ge \gamma _{\rm th} \right) f_n ,
\end{align}
where
\begin{align}\label{f_n}
f_n  = \left\{ \begin{array}{l}
 \Lambda^{-1}f_{\min} ,~~ if\;\hat F\left( \Lambda^{-1} f_{\min} \right) > \hat F\left(  \Lambda f_{\max}\right), \\
\Lambda f_{\max} ,\quad otherwise, \\
 \end{array} \right.
\end{align}
with $\Lambda  = 10^{\varepsilon /20}$, $f_{\min} = \min \left\{ {f_1^* , \cdots ,f_{n - 1}^* } \right\}$, $f_{\max} = \max \left\{ {f_1^* , \cdots ,f_{n - 1}^* } \right\}$,
and  $\hat{F}\left(f\right)= f^{-2}\left(L-\frac{ L^3 \pi^2 f^2}{6 c^2}\left(1-\frac{f_{\rm co}^2}{2f^2}- \cos\theta\right)^2\right)$. The corresponding subchannel bandwidth is
\begin{align}\label{optimal_BW_n}
B_n^*  = \min \left\{2f_n^* \frac{{10^{\varepsilon /20}  - 1}}{{10^{\varepsilon /20}  + 1}},{B_{total}-\sum\limits_{b=1}^{n-1} {B_b } }\right\},
\end{align}
respectively.
\end{theorem}
\begin{proof}
See Appendix B.
\end{proof}

Based on \textbf{Theorem 1} and \textbf{Theorem 2}, we obtain a low-complexity solution of problem \eqref{SA_problem}, which is concluded in \textbf{Algorithm~\ref{algorithmic1}}.
\begin{algorithm}[!htp]
\caption{ Solution of Problem \eqref{SA_problem}}\label{algorithmic1}
\begin{algorithmic}[1]
\STATE Initialize the number of subchannels $n=1$.
\STATE \textbf{if} $n=1$,\textbf{ then} \\
\STATE  ~~~Calculate the center frequency $f_1^{*}$ and bandwidth $B_1^*  = 2f_1^* \frac{{10^{\varepsilon /20}  - 1}}{{10^{\varepsilon /20}  + 1}}$ of the first subchannel based on \textbf{Theorem 1}. \\
\STATE  \textbf{else} \\
\STATE ~~~\textbf{Repeat}
\STATE  ~~~~~~$n= n+1$.\\
\STATE  ~~~~~~Calculate the center frequency $f_n^{*}$ and bandwidth $B_n^*$ of the $n$-th subchannel based on \textbf{Theorem 2}.  \\
\STATE  ~~~\textbf{Until} $B_n^{*}=0$.
\STATE  ~~~$N \leftarrow n-1$.\\
\STATE  \textbf{end if}  \\
\hspace{-0.7cm} \textbf{Output:} Subchannel allocation$\left(f_n^{*},B_n^*\right), n=1,\cdots,N$.
\end{algorithmic}
\end{algorithm}

{ {\textbf{Proposition 1:}} The proposed \textbf{Algorithm~\ref{algorithmic1}} can achieve the minimum number of subchannels for maximizing the transmission rate under QoS constraint.
\begin{proof}
Suppose that there is one subchannel $\widetilde{n}$ with the center frequency $f_{\widetilde{n}}$ and bandwidth $B_{\widetilde{n}}$.  If the SNR $F_{\rm SNR}\left(f_{\widetilde{n}}\right) < F_{\rm SNR}\left(f_{N}^*\right)$ where $f_N^*$ is the center frequency of the $N$-th subchannel obtained from \textbf{Algorithm~\ref{algorithmic1}}, the subchannel $\widetilde{n}$  cannot be added since $N$ subchannels with larger signal strength at center frequencies have already been selected under the total bandwidth constraint; If $F_{\rm SNR}\left(f_{\widetilde{n}}\right) \geq F_{\rm SNR}\left(f_{N}^*\right)$, the subchannel $\widetilde{n}$ should be part of one subchannel determined by \textbf{Algorithm~\ref{algorithmic1}}. The reason is that under the total bandwidth constraint, there are only $N$ subchannels with different $F_{\rm SNR}\left(f_n^*\right)(n=1,\cdots,N)$ values that are greater than $F_{\rm SNR}\left(f_N^*\right)$  and meet $|F_{\rm SNR}\left(f_{n+1}^*\right)\sigma_o^2\left| {_{\rm dB} } \right.-F_{\rm SNR}\left(f_{n}^*\right)\sigma_o^2\left| {_{\rm dB} } \right.| >\varepsilon$, as indicated in \textbf{Theorem 1} and \textbf{Theorem 2}. Hence $F_{\rm SNR}\left(f_{\widetilde{n}}\right)$ value should belong to $[F_{\rm SNR}\left(f_{N}^*\right),F_{\rm SNR}\left(f_{1}^*\right)]$, which completes the proof.
\end{proof} }
{The above subchannel allocation is designed based on the leaky-wave antenna's radiation pattern and channel quality (namely same transmit PSD value for all the subchannels). As such, the proposed \textbf{Algorithm~\ref{algorithmic1}} provides an efficient approach to select the number of subchannels and their bandwidths for very large THz frequency bands under QoS constraint. In the following section, we proceed to improve the energy efficiency after accomplishing the subchannel allocation.}

\section{Energy Efficiency Enhancement}\label{sec:EE}
The prior section has shown that the number of subchannels $N$ and the bandwidth of each subchannel can be easily determined with the help of \textbf{Algorithm~\ref{algorithmic1}}. To reduce the power consumption in the THz systems with leaky-wave antennas, we seek to maximize the average EE with respect to the transmit PSDs $\left\{q_n\right\}$ of the subchannels. Based on Sections II-IV, the considered problem is formulated as
\begin{align}\label{EE_problem}
&\mathop {\max }\limits_{\bf{q}} \frac{1}{N}\sum\limits_{n=1}^N \frac{\log _2 \left( 1 + q_n \Xi_n  \right)}{q_n+q_c}   \\
&\mathrm{s.t.} ~\mathrm{\hat{C}1:}~0< {q_n }  \le q_{\max} ,\;\,\;\forall n,\nonumber \\
&~~\quad\mathrm{\hat{C}2:}~q_n \Xi_n \ge \gamma _{\rm th} ,\;\,\;\forall n, \nonumber
\end{align}
where ${\bf{q}}=\left[q_n\right]$, $\Xi_n=\frac{\widetilde{G}\left(f_n ,\theta\right)\ell \left( {r} \right)} {\sigma_o^2 }$, and $q_c$ is the PSD due to the hardware's power consumptions. Constraint $\mathrm{\hat{C}1}$ is the transmit PSD's feasible range with the maximum value $q_{\max}$, and $\mathrm{\hat{C}2}$ is the QoS constraint.  Problem \eqref{EE_problem} can be decomposed into $N$ subproblems:
\begin{align}\label{EE_subproblem}
&\mathop {\max }\limits_{q_n}  \frac{\log _2 \left( 1 + q_n \Xi_n \right)}{q_n+q_c}   \\
&\mathrm{s.t.} ~\mathrm{\hat{C}1},~~\mathrm{\hat{C}2}. \nonumber
\end{align}
Then, we have the following theorem:
\begin{theorem}\label{theorem_3}
The optimal transmit PSD of the $n$-th subchannel is given by
\begin{align}\label{optimal_CF_n}
 q_n^*=\left\{ \begin{array}{l}
 q_{\max } ,\quad if\;\;\widehat{F}_{{\rm{EE}}} \left( {q_{\max } } \right) \ge 0, \\
 \max \left\{ {q_o ,\frac{{\gamma _{\rm th} }}{{\Xi _n }}} \right\},\quad otherwise, \\
 \end{array} \right.
\end{align}
where $\widehat{F}_{\rm{EE}}\left(q\right) = \frac{q  + q_c }{1 + q \Xi _n }\Xi _n  - \ln \left(1 + q \Xi _n  \right)$, and $q_o$ with $\widehat{F}_{\rm{EE}}\left(q_o\right)=0$ can be easily obtained by using a one-dimension search in $q_o \in \left(0,q_{\max }\right]$  since $\widehat{F}_{\rm{EE}}\left(q\right)$ is a decreasing function.
\end{theorem}
\begin{proof}
See Appendix C.
\end{proof}
Thus, we obtain a closed-form power allocation solution for enhancing the EE.

\section{Simulation Results}\label{sec:simulation}
This section provides numerical results to validate our analysis and the efficiency of the proposed solution for subchannel allocation and EE enhancement. In the simulations, {the transmit PSD is $q_t=-71.76$dBm/Hz (namely the total transmit power is 1W in a 15GHz bandwidth)\footnote{ Higher transmit power can be achieved, for instance, the existing work~\cite{Taiyun_Chi2017} has demonstrated that a CMOS-based source signal power can be about $-40$dBm/Hz. In addition, it is shown in~\cite{D_headland2018} that leaky-wave antenna is better than the conventional phased array from the perspective of radiation efficiency.}}, the noise's PSD is $\sigma_o^2=-168$dBm/Hz{\footnote{ Note that some existing THz testbeds such as~\cite{Priyangshu2020} have shown the level of noise's PSD in certain THz frequency bands.}}, the inter-plate distance is $d=3.5$mm, $\xi=1$, the free-space path loss model is applied with the LoS path loss exponent $\eta _{\rm LoS}=2$~\cite{Joonas2019,Joonas2020}, and the reference distance $D=1$. Based on the 3GPP blockage model~\cite{3gpp_blockage}, the LoS probability function  $P_{\rm LoS}\left(r\right)$ with a distance $r$ is given by
\begin{align}\label{3gpp_block}
P_{\rm LoS}\left(r\right)=e^{ - r/a_1 }  + \left( 1 - e^{ - r/a_1 }  \right)\min \left(\frac{a_2 }{r},1 \right),
\end{align}
where $a_1=63$m and $a_2=18$m. The other simulation parameters are detailed
in the following simulation results. { In addition,  the Monte Carlo simulation results are obtained by averaging over $3\times10^4$ trials.}

\subsection{Average Transmission Rate Analysis}
In this subsection, we analyze the effects of different system parameters on the average transmission rate. The analytical results for average transmission rate are obtained from \eqref{final_expression_rate}.

\begin{figure}[htbp]
\centering
\includegraphics[width=2.9 in,height=2.4 in]{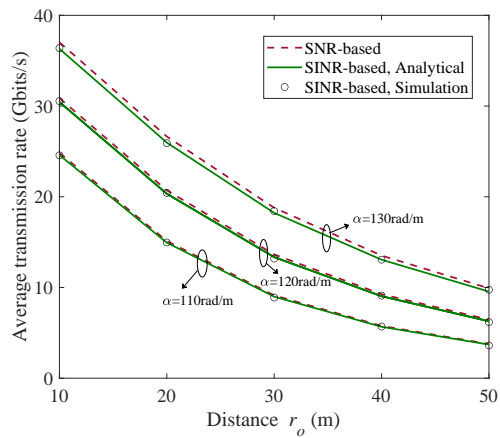}
\caption{The average transmission rate versus communication distance $r_o$ for different $\alpha$ with $\lambda_{\rm THz}=5*10^{-1}$/m$^2$, $L=0.06$m, $B_o=5$GHz and $f_o=270$GHz (namely the LoS direction of the typical receiver is $\theta_o=28.7^o$ based on \eqref{antenna_gain11}).}
\label{fig_distance}
\end{figure}

Fig.~\ref{fig_distance} shows that our analytical results  have a good match with the Monte Carlo simulations for different communication distance and attenuation coefficient values.  The use of leaky-wave antenna enables the THz networks to be noise-limited in the presence of highly dense transmitters. In the THz frequencies,  slightly increasing the communication distance between a typical transmitter and its receiver can significantly reduce the average transmission rate, due to the higher path loss. In addition, different attenuation coefficients has a negligible effect on the level of interference.

\begin{figure}[htbp]
\centering
\includegraphics[width=2.9 in,height=2.4 in]{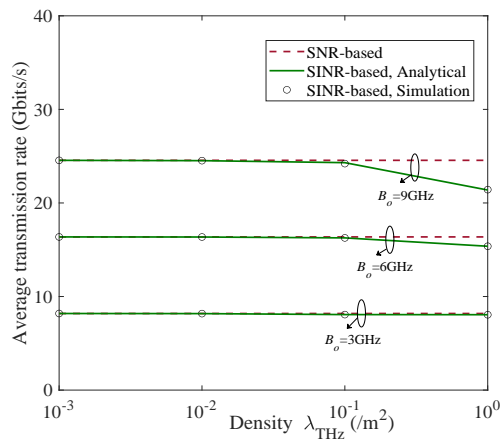}
\caption{The average transmission rate versus density $\lambda _{\rm{THz}}$ for different subchannel bandwidth $B_o$ with $r_o=30$m, $L=0.06$m, $\alpha=120$rad/m and $f_o=270$GHz(namely $\theta_o=28.7^o$).}
\label{fig_density}
\end{figure}

Fig.~\ref{fig_density} shows that THz networks with leaky-wave antennas become interference-limited only when extremely dense transmitters (e.g., $\lambda_{\rm THz}=10^{0}$/m$^2$ in this figure) utilize  large subchannel bandwidth. Moreover, we see that increasing the subchannel bandwidth results in a higher level of interference  as $\lambda_{\rm THz} > 10^{-1}$/m$^2$. The reason is that large subchannel bandwidth corresponds to higher propagation angle difference $\Delta \theta_o$ (See \eqref{bandwidth_radiation}), the effects of which are two-fold: 1) More interferers use the same frequency band; 2) The typical receiver is covered by more interferers.

{ Fig.~\ref{fig_fre} shows that the average transmission rate significantly decreases in higher center frequencies, due to the fact that spectral efficiency (bits/s/Hz) decreases in  higher frequencies with higher path losses under the same radiation pattern.} Again, we see that the analytical results match with the Monte Carlo simulations for different values of THz center frequency and aperture length. For higher THz frequencies (e.g., above 300GHz in this figure), the effect of aperture length on the interference is marginal.

\begin{figure}[htbp]
\centering
\includegraphics[width=2.9 in,height=2.4 in]{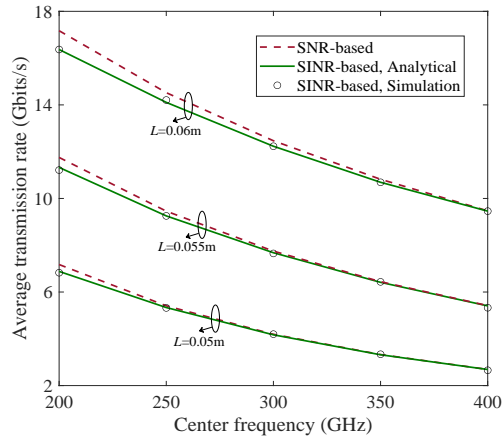}
\caption{The average transmission rate versus center frequency for different  $L$ with $r_o=30$m, $\lambda_{\rm THz}=5*10^{-1}$/m$^2$, $\alpha=120$rad/m and $B_o=5$GHz.}
\label{fig_fre}
\end{figure}

\subsection{Subchannel Allocation}\label{SA_optimal_sim}

 In this subsection, we focus on the efficiency of the proposed subchannel allocation solution through comparison with the same number of subchannels with the equal allocation of frequency band. In the simulations, the propagation angle from the transmitter to the receiver is uniformly distributed, i.e., $\theta \in U\left(0,\frac{\pi}{2}\right)$, the communication distance is uniformly distributed within the coverage radius $r_{\max}$,  the aperture length $L=0.06$m, $\gamma _{\rm th}=-6.5$dB~\cite{Sundeep2019}, $\varepsilon=0.2$dB, the continuous broadband spectrum ranging from 100 GHz to 350 GHz is considered, namely $f \in [100, 350]$(GHz). The proposed solution is provided based on \textbf{Algorithm}~\ref{algorithmic1}, which is in comparison with the equal allocation method (namely equal allocation of frequency band with the same number of subchannels and center frequency given by \eqref{antenna_gain11}).

\begin{figure}[htbp]
     \centering
    \subfigure[Average transmission rate]{
         \centering
         \includegraphics[width=2.9 in,height=2.4 in]{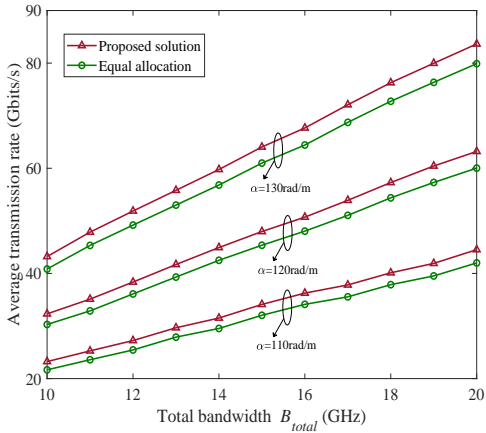}
      \label{fig1a}
     }
     \subfigure[Average number of subchannels]{
         \centering
         \includegraphics[width=2.9 in,height=2.4 in]{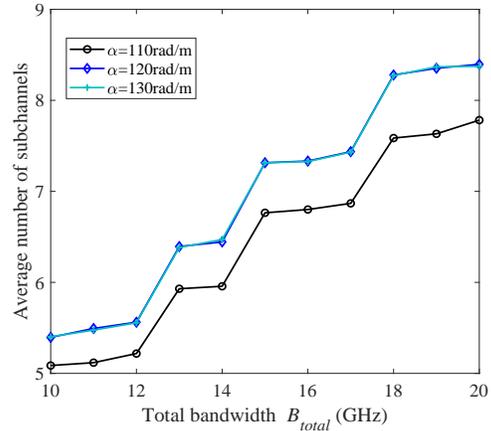}
      \label{fig1b}
    }
\caption{The average transmission rate and number of subchannels versus total bandwidth for different $\alpha$ with the coverage radius $r_{\max}=100$m.}
\label{fig_band}
\end{figure}

%\begin{figure}[htbp]
%     \centering
%    \subfigure[Average transmission rate]{
%         \centering
%         \includegraphics[width=2.5 in,height=2.1 in]{different_alpha.eps}
%      \label{fig3a}
%     }
%     \subfigure[Average number of subchannels]{
%         \centering
%         \includegraphics[width=2.5 in,height=2.1 in]{Numsub_different_alpha.eps}
%      \label{fig3b}
%    }
%%
%\caption{The average transmission rate and number of subchannels versus attenuation coefficient for different $B_{total}$ with the coverage radius $r_{\max}=100$m.}
%\label{fig_alpha}
%\end{figure}

\begin{figure}[htbp]
     \centering
    \subfigure[Average transmission rate]{
         \centering
         \includegraphics[width=2.9 in,height=2.4 in]{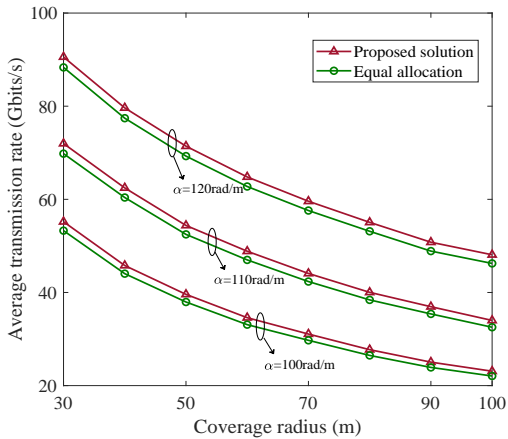}
      \label{fig2a}
     }
     \subfigure[Average number of subchannels]{
         \centering
         \includegraphics[width=2.9 in,height=2.4 in]{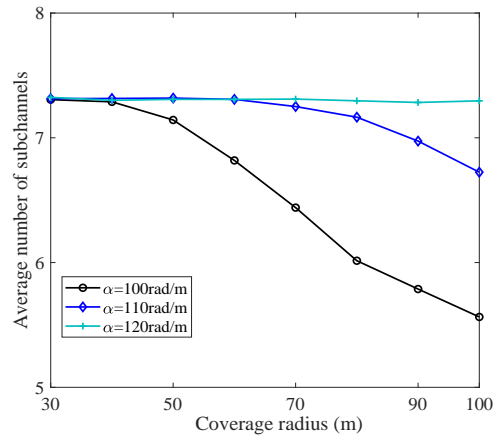}
      \label{fig2b}
    }
\caption{The average transmission rate and number of subchannels versus coverage radius for different $\alpha$ with the total bandwidth $B_{total}=15$GHz.}
\label{fig_cover}
\end{figure}

Fig.~\ref{fig_band} shows that the proposed solution outperforms the equal allocation method for different total bandwidths. Slightly increasing the $\alpha$ value leads to the higher average transmission rate when adding the frequency bandwidths. The reason is that based on \cite[eq. (7.25)]{Jackson_2008}, increasing the $\alpha$ value creates larger beamwidth, which means more frequencies are in the main-lobe that captures large radiated energy. It is seen from Fig. \ref{fig1b} that the proposed solution can well control the number of subchannels, e.g.,  in this figure, the number of subchannels increases by about $60\%$ when doubling the total bandwidth. There exist more subchannels for larger $\alpha$  and total bandwidth, due to the fact that more frequencies in the main-lobe satisfy the QoS constraint and can be applied.

Fig.~\ref{fig_cover} shows that the proposed solution achieves better performance than the equal allocation method for different coverage areas.  When the coverage area of the transmitter is expanded, the performance difference between the far-away receiver and the nearby one is significant, due to the higher path losses in the THz frequencies. An interesting phenomenon is seen in Fig.~\ref{fig2b}, i.e., the average number of subchannels decreases for larger coverage radius and lower attenuation coefficient values. The reason is that lower attenuation coefficient value creates narrower beamwidth, thus more frequencies are in the side-lobes (lower signal strength), which means that more frequencies cannot meet the QoS constraint.

%
%
%Fig.~\ref{fig_alpha} shows that the proposed solution achieves better performance than the equal allocation method for different attenuation coefficients. For a fixed total bandwidth, adding the attenuation coefficients can dramatically increase the average transmission rate, and the performance gap becomes larger for setting larger values of total bandwidth. It is seen from Fig. \ref{fig3b} that increasing the attenuation coefficient boosts the number of subchannels, and the number of subchannels becomes stable when $\alpha$ is large enough (e.g., $\alpha=120$rad/m in this figure) for a fixed total bandwidth. This is because that larger $\alpha$ results in many frequency bands with similar radiated energy, as mentioned before.

\begin{figure}[htbp]
     \centering
    \subfigure[Average transmission rate]{
         \centering
         \includegraphics[width=2.9 in,height=2.4 in]{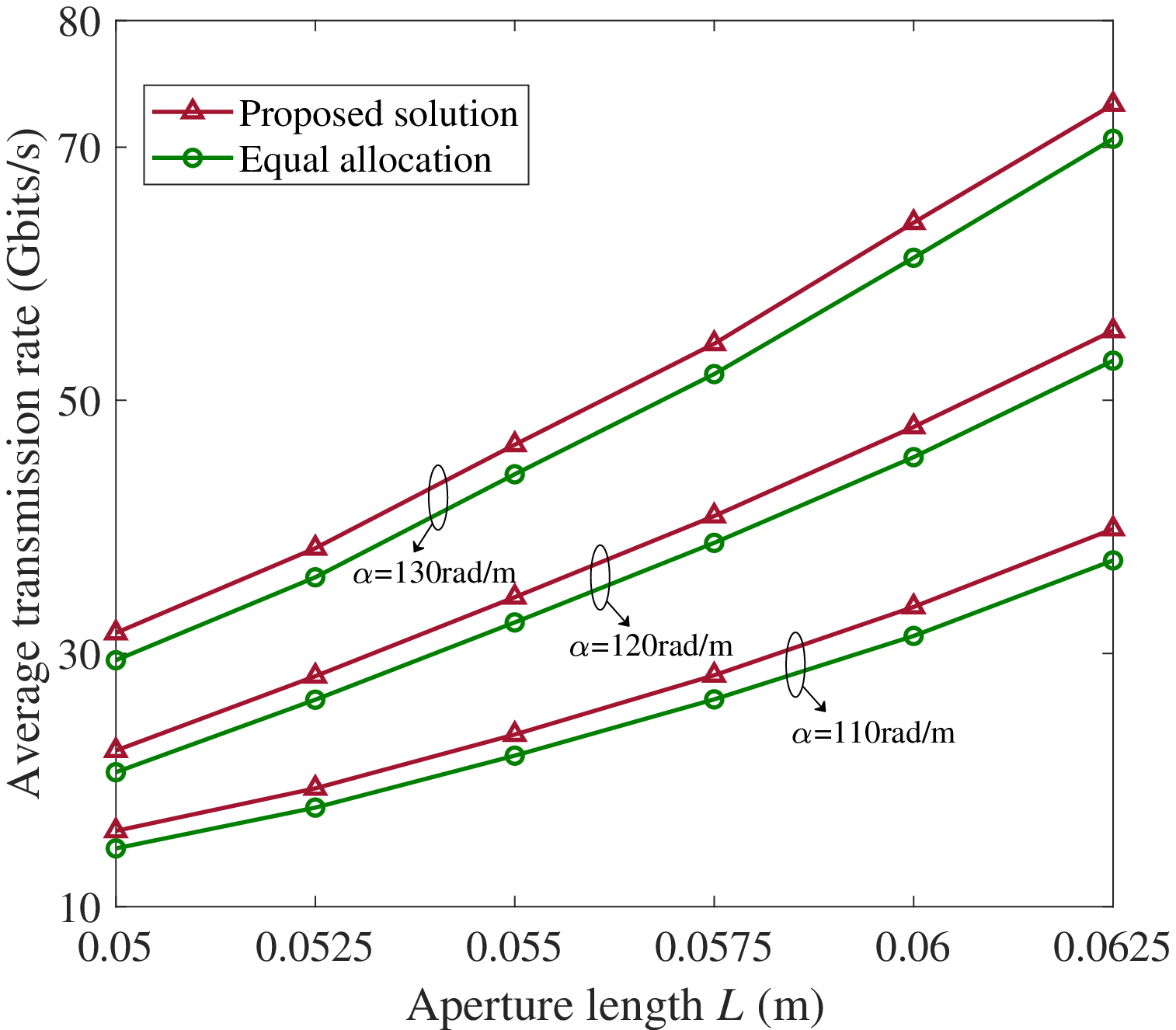}
      \label{fig4a}
     }
     \subfigure[Average number of subchannels]{
         \centering
         \includegraphics[width=2.9 in,height=2.4 in]{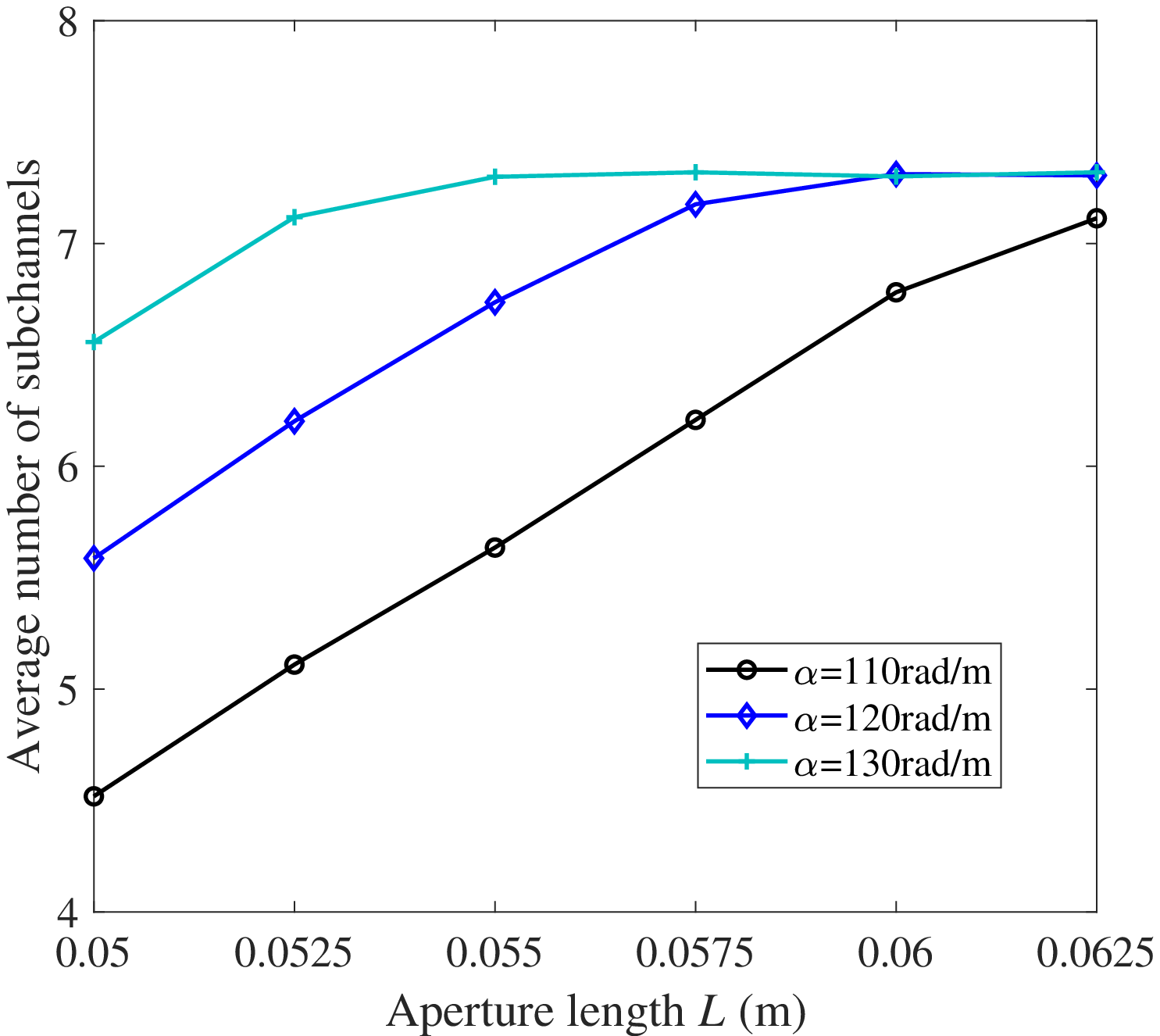}
      \label{fig4b}
    }
\caption{The average transmission rate and number of subchannels versus aperture length for different $\alpha$ with the coverage radius $r_{\max}=100$m and the total bandwidth $B_{total}=15$GHz.}
\label{fig_L}
\end{figure}

Fig.~\ref{fig_L} shows that the proposed solution achieves better performance than the equal allocation method for different aperture lengths. We see that slightly changing the aperture length  can have a big impact on the average transmission rate. Such an impact is more significant when increasing the $\alpha$ value. It is indicated from Fig.~\ref{fig4b} that slightly changing the aperture length also has a big impact on the number of subchannels.

%\begin{figure}[htbp]
%\centering
%\includegraphics[width=2.5 in,height=2.1 in]{EE_bandwidth.eps}
%\caption{The EE versus total bandwidth for different $\alpha$ with the aperture length $L=0.06$m.}
%\label{fig_EE_bandwidth}
%\end{figure}
%
%\begin{figure}[htbp]
%\centering
%\includegraphics[width=2.5 in,height=2.1 in]{EE_alpha.eps}
%\caption{The EE versus attenuation coefficient for different $L$ with the total bandwidth $B_{total}=15$GHz.}
%\label{fig_EE_alpha}
%\end{figure}
\subsection{Energy Efficiency}
{In this subsection, numerical results are presented  by using \textbf{Theorem}~\ref{theorem_3}, and the efficiency of the proposed power allocation solution is confirmed in comparison with the equal power allocation (namely $q_n=q_{\max}$, $n=1,\cdots,N$) under the same subchannel allocation obtained by using \textbf{Algorithm~\ref{algorithmic1}}}. In the simulations, $q_{\max}=-71.76$dBm/Hz, $q_{c}=-81.76$dBm/Hz,  the communication distance is uniformly distributed within the coverage radius $r_{\max}=100$m and other basic simulation parameters are the same as those mentioned in subsection~\ref{SA_optimal_sim}.
\begin{figure}[htbp]
     \centering
    \subfigure[The EE versus total bandwidth for different $\alpha$ with the aperture length $L=0.06$m.]{
         \centering
         \includegraphics[width=2.9 in,height=2.4 in]{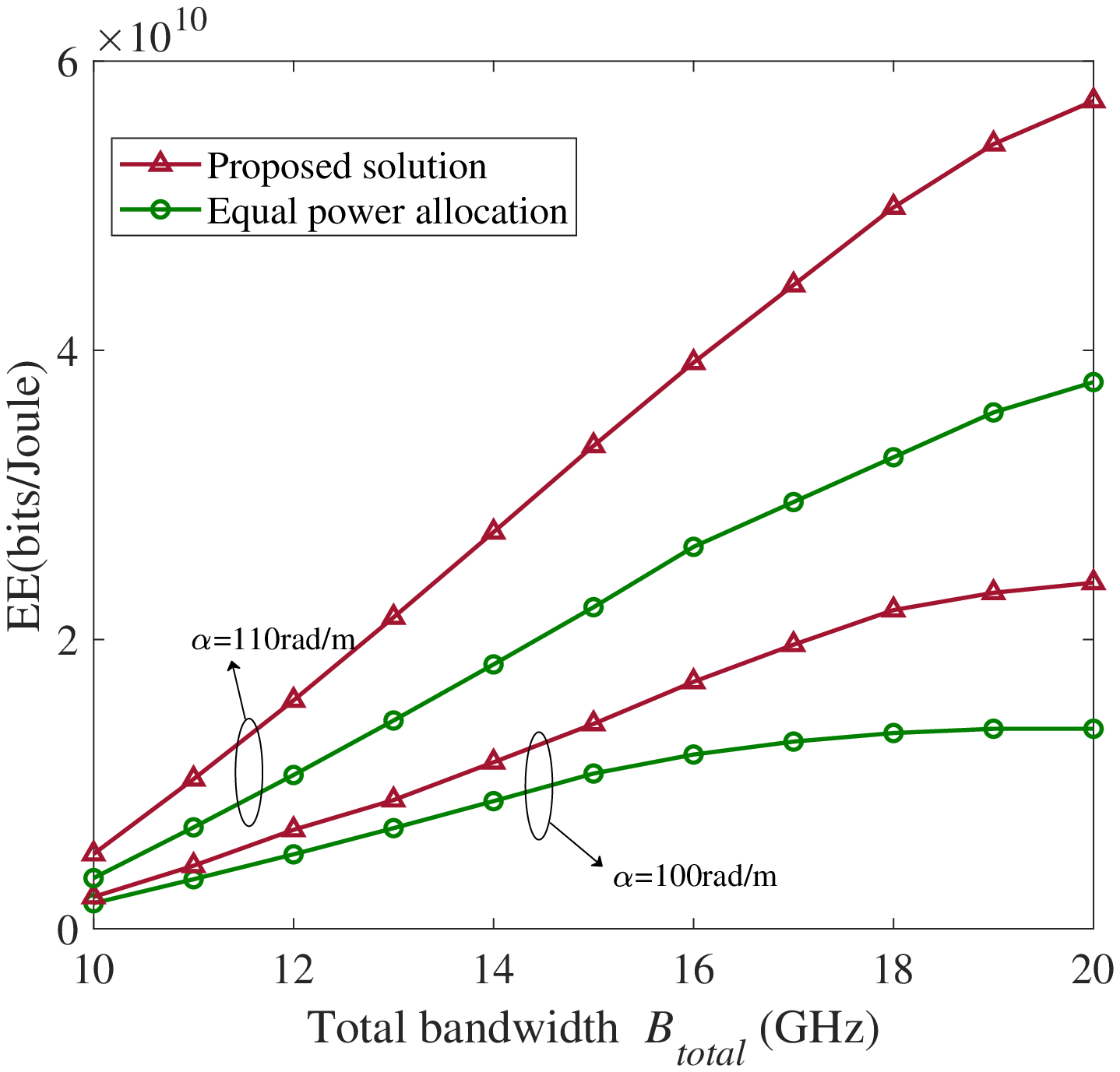}
      \label{fig_EE_bandwidth}
     }
     \subfigure[The EE versus attenuation coefficient for different $L$ with the total bandwidth $B_{total}=15$GHz.]{
         \centering
         \includegraphics[width=2.9 in,height=2.4 in]{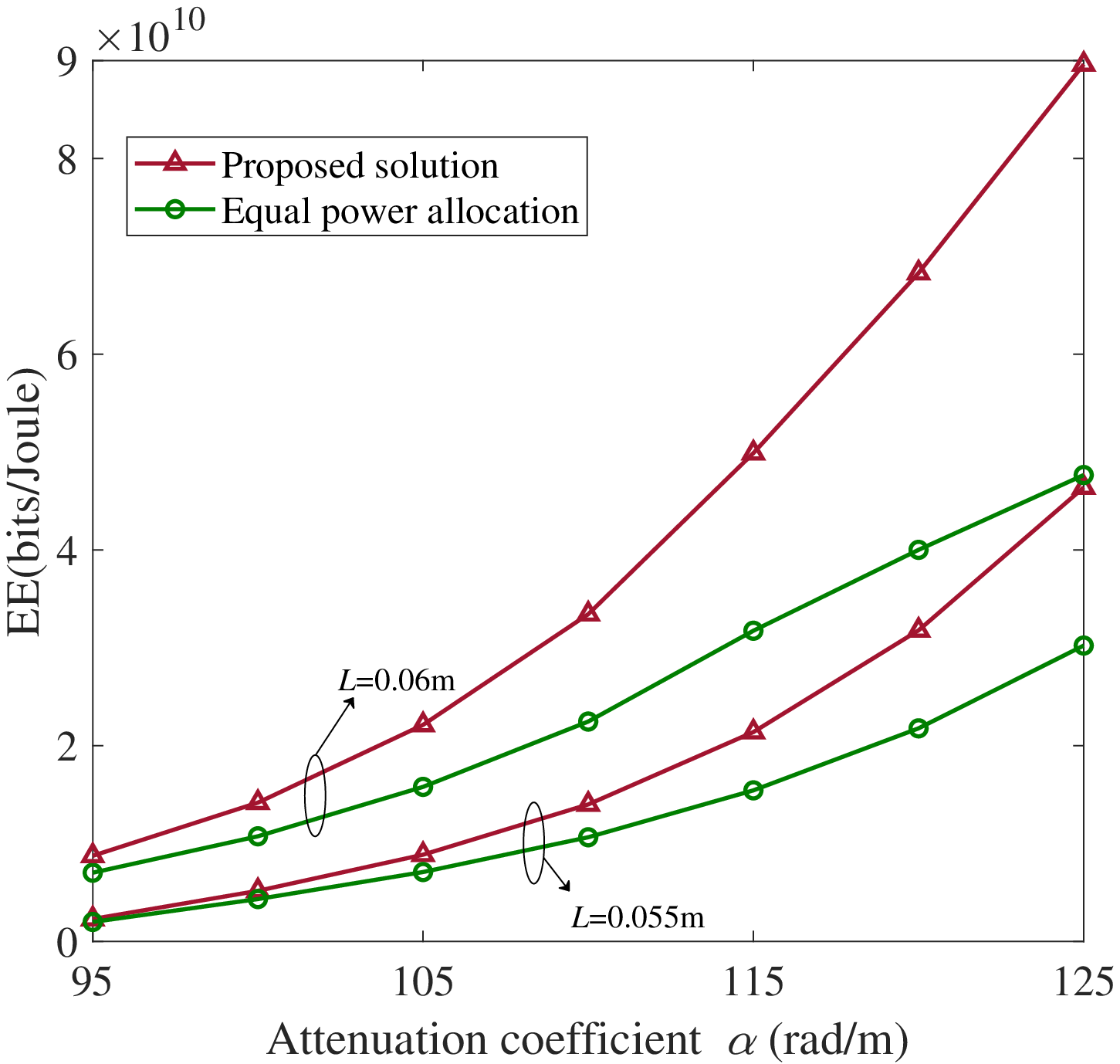}
      \label{fig_EE_alpha}
    }
\caption{Energy efficiency (EE) enhancement.}
\label{fig_ee}
\end{figure}

Fig.~\ref{fig_EE_bandwidth} shows that the proposed solution achieves better EE than the equal power allocation method for different total bandwidths, and the advantage of the proposed solution is more significant when increasing the total bandwidths or $\alpha$ value. The reason is that the increase of the total bandwidths or $\alpha$ value enables more available subchannels (See Fig.~\ref{fig_band}). The proposed solution can reduce more power consumption.
Fig.~\ref{fig_EE_alpha} shows that the proposed solution performs better than equal power allocation method for different attenuation coefficients. As mentioned before, the proposed solution saves more energy when increasing the $\alpha$ value. We see that slightly changing the aperture length has a big impact on the EE.

\section{Conclusions and Future Work}\label{conclusion_section}
This paper concentrated on the benefits of using the leaky-wave antenna in the THz networks, where each transmitter leveraged a single leaky-wave antenna to create high antenna gain. We first derived the average transmission rate in the dense THz networks and demonstrated the noise-limited behavior of using the leaky-wave antenna. Our results have shown that the effect of interference becomes significant only when the subchannel bandwidth is large enough and the transmitters are extremely dense. Then,  we addressed the subchannel allocation issue for large THz frequency bands, which is crucial for meeting the QoS constraint and tackling the PAPR issue. In light of the leaky-wave antenna's characteristics, a closed-form solution for subchannel allocation was developed. The results confirmed the efficiency of the proposed subchannel allocation solution compared to the equal allocation method. These results also indicated that the attenuation coefficient of the leaky-wave antenna has a substantial effect on the subchannel allocation. Furthermore, we developed a low-complexity power allocation method for EE enhancement, and the results showed that the proposed method achieves better EE performance than the equal power allocation.

While this work has shown the opportunities of using leaky-wave antennas in large-scale THz networks, more research efforts are required to further evaluate the performance behaviors from different perspectives and develop various transmission designs with leaky-wave antennas. In particular, the spatial-spectral feature of the leaky-wave antenna means that the value of the attenuation coefficient has to be properly designed. The reason is that larger attenuation coefficient values allow more frequencies to be in the main-lobe, which enables higher transmission rate. However, it also brings in more interference for the ultra-dense THz networks. Another area that warrants further research is the transmission design in the scenarios where each transmitter has multiple leaky-wave antennas. In this case, multiplexing gains are more likely to be achieved in the frequency-domain. But it may be hard to achieve large array gains due to the leaky-wave antenna's spatial-spectral feature, which differs from the beamforming/precoding designs with conventional large arrays in the sub-6 GHz and mmWave frequencies. In addition, further studies are needed for the scenarios such as multi-user transmissions, multi-tier transmissions with sub-6 GHz, mmWave and THz tiers,  cognitive radio, wiretap channels, integrated access and backhaul etc.

\section*{Appendix A: Proof of Theorem 1}
\label{App:theo_1}
\renewcommand{\theequation}{A.\arabic{equation}}
\setcounter{equation}{0}

Since $\ell\left( r_o \right)=(\frac{{\text{c}}}{{4\pi {f_1}}})^2\left(\max\left(D, r_o\right)\right)^{-\eta}$, problem \eqref{SA_problem_re} is equivalently transformed as
\begin{align}\label{A_1}
&\mathop {\max }\limits_{f_1}  \hat{F}\left(f_1\right)   \\
&\mathrm{s.t.} ~\mathrm{\widetilde{C}3}, \nonumber
\end{align}
where
\begin{align} \label{A_2}
\hat{F}\left(f_1\right)= f_1^{-2}\left(L-\frac{ L^3 \pi^2 f_1^2}{6 c^2}\left(1-\frac{f_{\rm co}^2}{2f_1^2}- \cos\theta\right)^2\right).
\end{align}
Taking the first-order and second-order derivatives of $\hat{F}\left(f_1\right)$ with respect to $f_1$  yields
\begin{align} \label{A_3}
\frac{{\partial \hat F}}{{\partial f_1 }} =  - {\rm{2}}f_1^{ - {\rm{3}}} \left( {L + \frac{{L^3 \pi ^2 f_{{\rm{co}}}^2 }}{{6c^2 }}\left( {1 - \frac{{f_{{\rm{co}}}^2 }}{{2f_1^2 }} - \cos \theta} \right)} \right),
\end{align}
and
\begin{align} \label{A_4}
{\frac{{\partial ^2 \hat F}}{{\partial f_1 ^2 }}}={\rm{6}}f_1^{ - 4} \left( {L + \frac{{L^3 \pi ^2 f_{{\rm{co}}}^2 }}{{6c^2 }}\left( {1 - \cos \theta} \right)} \right) - {\rm{5}}f_1^{ - 6} \frac{{L^3 \pi ^2 f_{{\rm{co}}}^4 }}{{6c^2 }},
\end{align}
respectively. Then, the solutions of $\frac{{\partial \hat F}}{{\partial f_1 }} = 0$ and ${\frac{{\partial ^2 \hat F}}{{\partial f_1 ^2 }}}= 0$ are given by
\begin{align} \label{A_5}
f_1^{(1)}=
\frac{{f_{{\rm{co}}}^2 }}{{\sqrt {\frac{{12c^2 }}{{L^2 \pi ^2 }} + 2\left( {1 - \cos \theta } \right)f_{{\rm{co}}}^2 } }},
\end{align}
and
\begin{align} \label{A_6}
f_1^{(2)}=\frac{{\sqrt {\rm{5}} f_{{\rm{co}}}^2 }}{{\sqrt {\frac{{{\rm{36}}c^2 }}{{L^2 \pi ^2 }} + 6 \left( {1 - \cos \theta} \right)f_{{\rm{co}}}^2} }},
\end{align}
respectively. According to \eqref{A_4} and \eqref{A_6}, we see that ${\frac{{\partial ^2 \hat F}}{{\partial f_1 ^2 }}} <0$ as $f_1 \in (0,f_1^{(2)})$, and ${\frac{{\partial ^2 \hat F}}{{\partial f_1 ^2 }}} >0$ as $f_1 \in (f_1^{(2)}, \infty)$. Since $f_1^{(1)} < f_1^{(2)}$, we have $\frac{{\partial \hat F}}{{\partial f_1 }} > \frac{{\partial \hat F}}{{\partial f_1 }}|_{f_1=f_1^{(1)}} =0 $ as $f_1 \in (0,f_1^{(1)})$, and $\frac{{\partial \hat F}}{{\partial f_1 }} < 0 $ as $f_1 \in (f_1^{(1)},\infty)$. Therefore, $f_1^{(1)}$ is the optimal solution for maximizing the objective function of problem \eqref{A_1}. Considering \eqref{A_5} and the constraint $\mathrm{C3}$,  we obtain \textbf{Theorem 1}.

\section*{Appendix B: Proof of Theorem 2}
\label{App:theo_1}
\renewcommand{\theequation}{B.\arabic{equation}}
\setcounter{equation}{0}
As mentioned in Appendix A, frequencies are selected to maximize the $\hat{F}\left(f\right)$ given by \eqref{A_2}, and  $\hat{F}\left(f\right)$ is the increasing function of $f$ as $f \in (0,f_1^*)$, and the decreasing function of $f$ as $f \in (f_1^*, \infty)$. Let $f_{\min} = \min \left\{ {f_1^* , \cdots ,f_{n - 1}^* } \right\}$, if the center frequency $f_n$ of the $n$-th subchannel satisfies $f_n \in (0,f_1^*)$, it should meet the following condition
\begin{align} \label{B_1}
f_n  + \frac{B_n}{2} = f_{\min }  - \frac{B_{\min}}{2},
\end{align}
where $B_{\min}$ is the bandwidth of the subchannel with the center frequency $f_{\min }$. Based on \eqref{SA_problem_re}, we see that $B_n=2f_n \frac{{10^{\varepsilon /20}  - 1}}{{10^{\varepsilon /20}  + 1}}$, and $B_{\min}=2f_{\min } \frac{{10^{\varepsilon /20}  - 1}}{{10^{\varepsilon /20}  + 1}}$. Thus, \eqref{B_1} is rewritten as
\begin{align} \label{B_2}
f_n  = 10^{ - \varepsilon /20} f_{\min }.
\end{align}
Likewise, let $f_{\max} = \max \left\{ {f_1^* , \cdots ,f_{n - 1}^* } \right\}$, if the center frequency $f_n$ of the $n$-th subchannel satisfies $f_n \in (f_1^*, \infty)$, we have
\begin{align} \label{B_3}
f_n  = 10^{ \varepsilon /20} f_{\max}.
\end{align}
Based on \eqref{A_1}, \eqref{B_2} and \eqref{B_3}, we find that the optimal $f_n$  is $10^{ - \varepsilon /20} f_{\min }$ when $\hat F\left( 10^{ - \varepsilon /20} f_{\min} \right) > \hat F\left( 10^{ \varepsilon /20} f_{\max}\right)$, otherwise it is $10^{ \varepsilon /20} f_{\max }$. As such, we obtain \textbf{Theorem 2}.
\section*{Appendix C: Proof of Theorem 3}
\label{App:theo_1}
\renewcommand{\theequation}{C.\arabic{equation}}
\setcounter{equation}{0}
Let $F_{\rm EE}\left(q_n\right)= \frac{1}{\ln 2}\frac{\ln\left( {1 + q_n \Xi_n} \right)}{q_n+q_c}$, which is the objective function of problem \eqref{EE_subproblem}. Taking the first-order derivative of $F_{\rm EE}$ yields
\begin{align} \label{C_1}
\frac{\partial F_{\rm EE}}{\partial q_n}= \frac{1}{\ln 2}\frac{\widehat{F}_{\rm{EE}}\left(q_n\right) }{\left( {q_n  + q_c } \right)^2},
\end{align}
where
\begin{align} \label{C_2}
\widehat{F}_{\rm{EE}}\left(q_n\right) = \frac{q_n  + q_c }{1 + q_n \Xi _n }\Xi _n  - \ln \left(1 + q_n \Xi _n  \right).
\end{align}
Since $\frac{{\partial \widehat{F}_{{\rm{EE}}} }}{{\partial q_n }} =  - \frac{{\left( {q_c  + q_n } \right)\Xi _n^2 }}{{\left( {1 + q_n \Xi _n } \right)^2 }} < 0$, $\widehat{F}_{\rm{EE}}\left(q_n\right)$  is a decreasing function of $q_n$. Based on the constraints $\mathrm{\hat{C}1}$ and $\mathrm{\hat{C}2}$, we see that $\frac{\gamma _{\rm th}}{\Xi_n}<q_n<q_{\max}$. Then, two cases need to be considered as follows:
\begin{itemize}
  \item Case 1: When $\widehat{F}_{\rm{EE}}\left(q_{\max}\right)\geq 0$, $\widehat{F}_{\rm{EE}}\left(q_n\right) > 0$ and thus $\frac{\partial F_{\rm EE}}{\partial q_n} > 0$ for $q_n \in \left(\frac{\gamma _{\rm th}}{\Xi_n},q_{\max}\right]$, i.e., $F_{\rm EE}\left(q_n\right)$ is an increasing function of $q_n$, hence the optimal $q_n^*$ is $q_n^*=q_{\max}$.
  \item Case 2: When $\widehat{F}_{\rm{EE}}\left(q_{\max}\right)<0$, $q_o$ with $\widehat{F}_{\rm{EE}}\left(q_o\right)=0$ can be easily obtained by using a one-dimension search for $q_o \in \left(0,q_{\max}\right]$. Then, we see that $\frac{\partial F_{\rm EE}}{\partial q_n} \geq 0$ for $q_n \in \left(0,q_o\right]$ and $\frac{\partial F_{\rm EE}}{\partial q_n} < 0$ for $q_n \in \left(q_o,q_{\max}\right]$. Therefore, $q_n^*=\max\left\{q_o,\frac{\gamma _{\rm th}}{\Xi_n}\right\}$ is the optimal solution.
\end{itemize}
As such, we get the optimal $q_n^*$ given in \textbf{Theorem}~\ref{theorem_3}.

\bibliographystyle{IEEEtran}

\end{document}